\documentclass[10pt,a4paper]{article}
\usepackage{amssymb,amsmath}
\usepackage{amsthm}
\usepackage{bbm}
\usepackage{stmaryrd}
\usepackage[usenames,dvipsnames]{xcolor}
\usepackage{color}

\input xy
\xyoption{all} %\tolerance=500

\newtheorem{theorem}{Theorem}[section]
\newtheorem{corollary}[theorem]{Corollary}

\newtheorem{proposition}[theorem]{Proposition}

\theoremstyle{definition}
\newtheorem{definition}[theorem]{Definition}
\newtheorem{remark}[theorem]{Remark}
\newtheorem{example}[theorem]{Example}

\newenvironment{warning}[1][Warning.]{\begin{trivlist}
\item[\hskip \labelsep {\bfseries #1}]}{\end{trivlist}}

\newcommand{\Dleft}{[\hspace{-1.5pt}[}
\newcommand{\Dright}{]\hspace{-1.5pt}]}
\newcommand{\SN}[1]{\Dleft #1 \Dright}

\newcommand{\Id}{\mathbbmss{1}}

\newcommand*{\longhookrightarrow}{\ensuremath{\lhook\joinrel\relbar\joinrel\rightarrow}}
\newcommand*{\rRelation}{\ensuremath{\joinrel\relbar\joinrel\joinrel\relbar\joinrel \vartriangleright}}

\DeclareMathOperator{\Vect}{Vect}

\DeclareMathOperator{\Span}{Span}

\DeclareMathOperator{\w}{w}

\font\black=cmbx10 \font\sblack=cmbx7 \font\ssblack=cmbx5 \font\blackital=cmmib10  \skewchar\blackital='177
\font\sblackital=cmmib7 \skewchar\sblackital='177 \font\ssblackital=cmmib5 \skewchar\ssblackital='177
\font\sanss=cmss10 \font\ssanss=cmss8 %scaled 900
\font\sssanss=cmss8 scaled 600 \font\blackboard=msbm10 \font\sblackboard=msbm7 \font\ssblackboard=msbm5
\font\caligr=eusm10 \font\scaligr=eusm7 \font\sscaligr=eusm5  
  
\font\bsymb=cmsy10 scaled\magstep2
\def\all#1{\setbox0=\hbox{\lower1.5pt\hbox{\bsymb
       \char"38}}\setbox1=\hbox{$_{#1}$} \box0\lower2pt\box1\;}
\def\exi#1{\setbox0=\hbox{\lower1.5pt\hbox{\bsymb \char"39}}
       \setbox1=\hbox{$_{#1}$} \box0\lower2pt\box1\;}

\def\tx#1{{\fam0\relax#1}}

\newfam\bifam
\textfont\bifam=\blackital \scriptfont\bifam=\sblackital \scriptscriptfont\bifam=\ssblackital

\newfam\blfam
\textfont\blfam=\black \scriptfont\blfam=\sblack \scriptscriptfont\blfam=\ssblack

\newfam\bbfam
\textfont\bbfam=\blackboard \scriptfont\bbfam=\sblackboard \scriptscriptfont\bbfam=\ssblackboard

\newfam\ssfam
\textfont\ssfam=\sanss \scriptfont\ssfam=\ssanss \scriptscriptfont\ssfam=\sssanss
\def\sss#1{{\fam\ssfam\relax#1}}

\newfam\clfam
\textfont\clfam=\caligr \scriptfont\clfam=\scaligr \scriptscriptfont\clfam=\sscaligr

\def\hpb#1{\setbox0=\hbox{${#1}$}
    \copy0 \kern-\wd0 \kern.2pt \box0}
\def\vpb#1{\setbox0=\hbox{${#1}$}
    \copy0 \kern-\wd0 \raise.08pt \box0}

\def\pmb#1{\setbox0\hbox{${#1}$} \copy0 \kern-\wd0 \kern.2pt \box0}
\def\pmbb#1{\setbox0\hbox{${#1}$} \copy0 \kern-\wd0
      \kern.2pt \copy0 \kern-\wd0 \kern.2pt \box0}
\def\pmbbb#1{\setbox0\hbox{${#1}$} \copy0 \kern-\wd0
      \kern.2pt \copy0 \kern-\wd0 \kern.2pt
    \copy0 \kern-\wd0 \kern.2pt \box0}
\def\pmxb#1{\setbox0\hbox{${#1}$} \copy0 \kern-\wd0
      \kern.2pt \copy0 \kern-\wd0 \kern.2pt
      \copy0 \kern-\wd0 \kern.2pt \copy0 \kern-\wd0 \kern.2pt \box0}
\def\pmxbb#1{\setbox0\hbox{${#1}$} \copy0 \kern-\wd0 \kern.2pt
      \copy0 \kern-\wd0 \kern.2pt
      \copy0 \kern-\wd0 \kern.2pt \copy0 \kern-\wd0 \kern.2pt
      \copy0 \kern-\wd0 \kern.2pt \box0}

\mathchardef\za="710B  %\alpha
\mathchardef\zb="710C  %\beta
\mathchardef\zg="710D  %\gamma
\mathchardef\zd="710E  %\delta
\mathchardef\zve="710F %\epsilon
\mathchardef\zz="7110  %\zeta
\mathchardef\zh="7111  %\eta
\mathchardef\zvy="7112 %\theta
\mathchardef\zi="7113  %\iota
\mathchardef\zk="7114  %\kappa
\mathchardef\zl="7115  %\lambda
\mathchardef\zm="7116  %\mu
\mathchardef\zn="7117  %\nu
\mathchardef\zx="7118  %\xi
\mathchardef\zp="7119  %\pi
\mathchardef\zr="711A  %\rho
\mathchardef\zs="711B  %\sigma
\mathchardef\zt="711C  %\tau
\mathchardef\zu="711D  %\upsilon
\mathchardef\zvf="711E %\phi
\mathchardef\zq="711F  %\chi
\mathchardef\zc="7120  %\psi
\mathchardef\zw="7121  %\omega
\mathchardef\ze="7122  %\varepsilon
\mathchardef\zy="7123  %\vartheta
\mathchardef\zf="7124  %\varomega
\mathchardef\zvr="7125 %\varrho
\mathchardef\zvs="7126 %\varsigma
\mathchardef\zf="7127  %\varphi
\mathchardef\zG="7000  %\Gamma
\mathchardef\zD="7001  %\Delta
\mathchardef\zY="7002  %\Theta
\mathchardef\zL="7003  %\Lambda
\mathchardef\zX="7004  %\Xi
\mathchardef\zP="7005  %\Pi
\mathchardef\zS="7006  %\Sigma
\mathchardef\zU="7007  %\Upsilon
\mathchardef\zF="7008  %\Phi
\mathchardef\zW="700A  %\Omega
\mathchardef\zC="7009  %\Psi

\newcommand{\be}{\begin{equation}}
\newcommand{\ee}{\end{equation}}

\newcommand{\bea}{\begin{eqnarray}}
\newcommand{\eea}{\end{eqnarray}}
\newcommand{\beas}{\begin{eqnarray*}}
\newcommand{\eeas}{\end{eqnarray*}}
\def\*{{\textstyle *}}
\newcommand{\R}{{\mathbb R}}

\newcommand{\D}{{\rm d}}

\newcommand{\we}{\wedge}
\newcommand{\nn}{\nonumber}

\newcommand{\s}{{\textstyle *}}

\newcommand{\pa}{\partial}
\newcommand{\ti}{\times}
\newcommand{\A}{{\cal A}}

\def\sP{{\sss P}}

\def\sT{{\sss T}}

\def\sj{{\sss j}}

\def\xi{\tx{i}}

\def\s*{{\scriptstyle *}}

\begin{document}

\author{\\
        Andrew James Bruce$^1$\\ Katarzyna  Grabowska$^2$\\
        Janusz Grabowski$^1$\\
        \\
         $^1$ {\it Institute of Mathematics}\\
                {\it Polish Academy of Sciences }\\
         $^2$ {\it Faculty of Physics}\\{\it University of Warsaw} }

\date{\today}
\title{Linear duals of graded bundles\\ and higher analogues of (Lie) algebroids\thanks{Research funded by the  Polish National Science Centre grant under the contract number DEC-2012/06/A/ST1/00256.  }}
\maketitle

\begin{abstract}
\emph{Graded bundles} are  a class of graded manifolds which represent  a  natural generalisation of vector bundles  and include the higher order tangent bundles as canonical examples. We present and study the concept of the \emph{linearisation} of graded bundle which allows us to define the notion of the \emph{linear dual} of a graded bundle. They are examples of double structures, \emph{graded-linear} ($\mathcal{GL}$) \emph{bundles}, including double vector bundles as a particular case. On $\mathcal{GL}$-bundles we define what we shall call \emph{weighted algebroids}, which are to be understood as algebroids in the category of graded bundles. They can be considered as a geometrical framework for higher order Lagrangian mechanics. Canonical examples are reductions of higher tangent bundles of Lie groupoids. Weighted algebroids  represent also a generalisation of \emph{$\mathcal{VB}$-algebroids} as defined by Gracia-Saz \&  Mehta and the \emph{$\mathcal{LA}$-bundles} of Mackenzie. The resulting structures are strikingly similar to Voronov's higher Lie algebroids, however our approach does not require the initial  structures to be defined on supermanifolds.
\end{abstract}



\section{Introduction}\label{sec:Intro}
 Manifolds and supermanifolds that carry various additional gradings on their structure sheaves are now an established part of mathematics with a long history. The general theory of graded manifolds in our understanding was initiated by Voronov \cite{Voronov:2001qf}.  An important class of such supermanifolds are those that carry a non-negative integer grading, especially \emph{N-manifolds} (see for example  Roytenberg \cite{Roytenberg:2001} and \v{S}evera \cite{Severa:2001}).

In the purely even setting, a vector bundle is completely encoded in its homogeneous structure. In particular, Grabowski and Rotkiewicz showed in \cite{Grabowski2009} that a vector bundle structure on a manifold is equivalent to a regular action of the multiplicative monoid of reals on that manifold (regular \emph{homogeneity structure}). This also lead to a natural and simple definition of a double vector bundle as a manifold equipped with a pair of commuting {homogeneity structures}. Relaxing the regularity condition leads to the concepts of \emph{graded bundles} and \emph{$n$-tuple graded bundles}, which were studied by Grabowski and Rotkiewicz \cite{Grabowski2012}. They proved that a manifold with a smooth action of the multiplicative monoid of reals is a graded bundle in a traditional understanding, i.e. homogeneous local coordinates can always be found. This is in the spirit of \v{S}evera's definition of N-manifold given in \cite{Severa:2001}, however without proving such an equivalence.
We should stress that we speak about \emph{graded bundle} rather than \emph{graded manifold} in order to be precise and make a sharp distinction with other uses of the latter term in the literature, including the one related to a $\mathbb{Z}_2$-gradation (or just supermanifolds).

 A little more specifically, a vector bundle structure $E \rightarrow M$  is encoded by  assigning a weight of zero to the base coordinates and one to the linear coordinates on the total space. Thus, there is essentially a  one-to-one  correspondence between vector bundles and  graded bundles for which we can assign the weight zero and one. The condition of the weight to be  one for the fibre coordinates is a restatement of linearity. Thus philosophically, a graded bundle should be viewed as a ``non-linear or higher vector bundle" via the relaxing of the assignment of weight.

This philosophy also allows a higher order generalisaton of a Lie algebroid to appear. Recall that a Lie algebroid structure on $E \rightarrow M$ can be viewed as a  weight one homological vector field on the supermanifold $\Pi E$  \cite{Vaintrob:1997}. Here $\Pi$ is the parity reversion functor and it acts by shifting the Grassmann parity of the linear coordinates relative to the linear coordinates on $E$. The weight is assigned as previously explained. By \emph{homological} we mean an odd vector field that self-commutes. Voronov proposed that a non-linear or higher Lie algebroid should be understood as a homological vector field of weight one on a non-negatively graded supermanifold  \cite{Voronov:2001qf,voronov-2010}. It is important to note that Voronov's definition requires us to start from the underlying  structure of a supermanifold; without super it would be impossible to obtain a homological vector field. This, we believe,  has hampered the understanding of examples and potential applications of non-linear Lie algebroids   in  standard (even) differential geometry and geometric mechanics.

Recall that a Lie algebroid structure can also be viewed as a linear Poisson structure on the total space of the dual bundle $E^{*}$. In the graded language one can say we have a  weight minus one Poisson structure. Furthermore, in \cite{Grabowski:1999} a Lie algebroid  structure (and some  natural generalisations there of) on $E$  was considered as a double vector bundle morphism
\begin{equation}
\epsilon : \sT^{*}E \rightarrow \sT E^{*},
\end{equation}
\noindent covering the identity on $E^{*}$.

It is  mainly this approach to Lie algebroids as double vector bundle morphisms and linear Poisson structures that we generalise in this paper to construct a higher analogue of a general algebroid, although we also provide  other descriptions. The main reason for taking certain maps as our definition of a general algebroid (say skew or Lie) is that it is precisely these maps that are fundamental in geometric mechanics; we develop these ideas in the graded setting in \cite{Bruce:2014b} and apply them to higher order mechanics. In particular, the various Lie-type brackets associated with a Lie algebroid  will be viewed as secondary or derived notions rather than fundamental. For an overview of the various bracket operations found in algebra, geometry and mechanics the reader can consult \cite{Grabowski:2013a}. Similarly, we recognise and appreciate the power of supermanifolds in describing Lie algebroids, though from the point of view of geometric mechanics, homological vector fields do not shed light on formulating Lagrangian and Hamiltonian mechanics.

Another approach to higher analogues of Lie algebroids  based on  generalising the \emph{kappa-relation} $\kappa : \sT E \rRelation \sT E$, which is dual to the double vector bundle morphism $\epsilon$,  has recently  been explored by J\'{o}\'{z}wikowski \& Rotkiewicz \cite{Jozwikowski:2014}. Their basic idea is that reductions of higher tangent bundles of Lie groupoids should be examples of ``higher algebroids". They also explore the dual map in their version of a higher algebroid to tackle higher order Lagrangian mechanics. However, their approach is essentially based on ``lifting"  an almost Lie algebroid structure on a vector bundle to a particular canonically constructed graded bundle. The constructions are somewhat complicated and the interested reader is directed towards the original literature.

As graded bundles are not just vector bundles, a little more work has to be done in constructing their duals. Indeed, the notion of an object dual to a graded bundle has been an open question and here we propose a (partial) solution.  In order to do this we employ a \emph{canonical linearisation} process to obtain from a graded bundle a double graded bundle which has as vector bundle structure as one of the graded subbundles. We will properly define the linearisation of a graded bundle  in due course. However, for now note that the linearisation  and the linear dual of a graded bundle are particular reductions of the tangent and cotangent bundle, respectively,  of the graded bundle in question  (see Section \ref{sec:associating double}). Thus, both the linearisation and the linear dual are naturally double graded bundles.

With the notion of a linear dual of a graded bundle in place one can more or less follow standard constructions to define a ``graded algebroid", one only needs to be careful with the extra weights inherited from the initial graded bundle and the linearisation. As the  linearisation of  a graded bundle is a reduction of its tangent bundle, \emph{weighted algebroids}  as we shall call them, are  particular morphisms of certain  triple graded bundles. In fact, we are able to define and study a slightly wider class of objects by first defining $\mathcal{GL}$-bundles, which are a natural abstraction of the tangent bundle of a graded bundle. In particular, $\mathcal{GL}$-bundles have a vector bundle structure as one of the graded subbundles, similar to the linearisation of a graded bundle.

Weighted Lie algebroids can  be understood as particular graded Poisson structures on the dual of the underlying $\mathcal{GL}$-bundle. As $\mathcal{GL}$-bundle have a linear structure, taking the dual and applying the parity  function make sense. Importantly, we can pass from such a graded Poisson structure to a  homological vector field on the parity shifted version of the  $\mathcal{GL}$-bundle of weight $(0,1)$, or simply total weight one.  Thus, weighted   Lie algebroids give class of objects very similar to Voronov's higher Lie algebroids.  We stress at this stage that the (bi-)weight of the associated homological vector field is not a forced condition, but rather it follows very naturally from the initial definitions in terms of graded morphisms.  Thus, much like the classical case of Lie algebroids, we can start from purely even manifolds.  This opens up canonical examples and potential applications.  In particular the tangent bundle of a graded bundle and the linearisation of the k-th order tangent bundle of a manifold are both canonical examples of a weighted  Lie algebroids.  Moreover, we shall show that the reduction of a higher order tangent bundle of a Lie groupoid, via its linearisation, also comes equipped canonically with the structure of a weighted algebroid. Interestingly, the weighted Lie algebroid associated with the higher order tangent bundle of a Lie groupoid is an example of what is known in the literature as a \emph{prolongation of a fibred manifold with respect to a Lie algebroid} (cf. \cite{Cortes:2009} and references therein).  We shall comment on this further in Section \ref{sec:examples}.

We also note that weighted algebroids should be viewed as a higher generalisation of $\mathcal{VB}$-algebroids as defined by Gracia-Saz \&  Mehta \cite{Gracia-Saz:2009}, who also showed that they are equivalent to  Mackenzie's $\mathcal{LA}$-vector bundles \cite{Mackenie:1998}. $\mathcal{VB}$-algebroids are also studied in  recent papers by Brahic,  Cabrera \& Ortiz \cite{Brahic:2014} and Bursztyn, Cabrera \& del Hoyo \cite{Bursztyn:2014}. Loosely, $\mathcal{VB}$-algebroids are ``Lie algebroids in the category of vector bundles", while  $\mathcal{LA}$-vector bundles are ``vector bundles in the category of Lie algebroids".  These  ``double objects" are double vector bundles with additional structure and can be thought of as generalised Lie algebroid representations. Mackenzie's $\mathcal{LA}$-vector bundles were originally defined while working towards the notion of a double Lie algebroid. As a remark, the original definition of a double Lie algebroid by Mackenzie is rather complicated and was reformulated and vastly simplified  Voronov \cite{Voronov:2012}  in terms of a pair of commuting homological vector fields on a (super) double vector bundle carrying wight $(1,0)$ and $(0,1)$. Grabowski \& Rotkiewicz \cite{Grabowski2009} further simplified  this picture by realising that the structure  of double Lie algebroids is really encoded in four vector fields by encoding the double vector bundle structure in a pair of commuting (h-complete) Euler vector fields.

Weighted algebroids can be viewed as particular examples of double graded bundles, in particular we require that one of the graded subbundles be a vector bundle,  that carry an additional structure.  We can then think of weighted Lie algebroids as ``graded bundles in the category of Lie algebroids" or equivalently as ``Lie algebroids in the category of graded bundles".  From this perspective our approach is a mixture of  higher and multiple. However,  we will not make explicit use of the categorical descriptions of  ``double objects" in the sense of Ehresmann in this paper.

We again remark that our main motivation for this work was to understand what kind of higher versions of Lie algebroids can be found in the  theory of higher order mechanics  and field theory. We want to generalise the geometric approaches indicated for standard mechanics in \cite{Grabowska:2008,Grabowska:2006} and for higher-order Lagrangians in \cite{Grabowska2014,Grabowska2014a} to ``higher Lie algebroids''.  Extending the geometric tools of the Lagrangian formalism on tangent bundles to Lie algebroids was motivated by the fact that reductions usually push one out of the environment of tangent bundles. In a similar way, reductions of higher order tangent bundles will push one into the environment of ``higher Lie  algebroids". However, it is not immediately clear exactly what the most suitable and encompassing definition of a ``higher algebroid" is in the context of higher order mechanics. The concepts of Voronov \cite{voronov-2010} or J\'o\'zwikowski \& Rotkiewicz \cite{Jozwikowski:2014,Jozwikowski:2015} seem not to fit exactly to our needs. We  explore the applications of weighted algebroids in both the higher order Lagrangian and Hamiltonian formalisms in a separate publication \cite{Bruce:2014b}. The notion of a \emph{weighted groupoid} as a Lie groupoid carrying a compatible graded structure, and the  Lie theory relating such objects to weighted Lie algebroids is developed by us in \cite{Bruce:2014c}.

In this paper we present a systematic study of  the geometry of graded bundles, focusing on the pure and deep mathematical side of their linearisation and weighted algebroids. As we shall see, basic examples of weighted Lie algebroids arise in relation to graded Poisson structures on graded bundles, tangent bundles of graded bundles and the linearisation of higher order tangent bundles. From this perspective, weighted algebroids seem a very natural unifying concept in the theory of graded bundles. \smallskip

\noindent \textbf{Remarks on nomenclature.} The objects we define and study are algebroids that carry in addition to their standard weight, which is  associated with the underlying vector bundle structure, an extra independent  weight. One is tempted to simply refer to such objects as ``graded algebroids" however, our subsequent  notion of a ``graded Lie algebra" would not coincide with standard definitions. The nomenclature ``higher" or ``non-linear" already has specific meaning. Thus, we choose the well-suited nomenclature ``weighted algebroids". However, we must stress that this notion is not to be confused with the notion of a weight in the sense of the representation theory of Lie algebras and similar.\smallskip

\noindent \textbf{Our use of supermanifolds.} Although we formulate weighted algebroids initially in the fully commutative setting, it is clear that the framework of supermanifolds is very elegant and powerful in the context of algebroid-like objects. In particular we will make use of the description of multivector fields and the Schouten bracket in terms of functions on the antitangent bundle and it's canonical odd symplectic structure. Furthermore, non-trivial homological vector fields  are `odd', so can only be defined on supermanifolds. The physicist's correct and  intuitive definition of a supermanifold as a manifold with both commuting and anticommuting coordinates will suffice for this paper.  That said, technically  we will follow the ``Russian School"  and understand supermanifolds in terms of locally superringed spaces. There are now several good introductory books on the theory supermanifolds including \cite{Carmeli:2011,Deligne:1999,Manin1997,Varadarajan2004} which the reader can consult.\smallskip

\noindent \textbf{An after thought.} After we had written the first version of this paper, an arXiv preprint by Jotz Lean  appeared \cite{JotzLean:2015}. One key result of her work is the categorical equivalence between N-manifolds of degree $2$ and what she calls  \emph{metric double vector bundles}. Mod Grassmann parity, N-manifolds of degree $2$ are essentially the same as graded bundles of degree $2$. Thus, the categorical equivalence as described by Jotz Lean is implicitly included in our constructions, however it requires some work to make this explicit. As establishing categorical equivalences is not the goal of this paper, we postpone a careful comparison of the linearisation functor and the constructions of Jotz Lean to a separate paper \cite{Bruce:2015}.

\smallskip

\noindent \textbf{Arrangement of the paper.} In Section \ref{sec:graded bundles} we recall some of the basic theory of graded bundles,  $n$-tuple graded bundles and develop  some further technology to be used throughout the rest of this paper. In Section \ref{sec:associating double} we present details of the linearisation of a graded bundle and use this to define the notion of the linear dual of a graded bundle. In Section \ref{sec:weighted algebroids} we define the notion of weighted (skew and Lie) algebroids as particular morphisms of certain triple graded bundles.  In this section we also comment on how   weighted algebroids are related to Voronov's higher Lie algebroids and how they represent a generalisation of the $\mathcal{VB}$-algebroids of  Gracia-Saz \& Mehta, as well as their understanding by Bursztyn, Cabrera \& del Hoyo \cite{Bursztyn:2014}, and so the $\mathcal{LA}$-vector bundles of  Mackenzie.  In Section \ref{sec:examples} we present further  examples of weighted Lie algebroids  and in particular we show how they arise via reduction and linearisation of higher tangent bundles of Lie groups and Lie groupoids.

\section{Graded and $n$-tuple graded bundles}\label{sec:graded bundles}
\subsection{Graded manifolds and graded bundles}
The general theory of graded manifolds as employed in this work was developed  by Voronov  in $\cite{Voronov:2001qf}$. Following the standard notion of a supermanifold, a graded manifold is defined in terms of its structure sheaf which is $\mathbb{Z}\times \mathbb{Z}_{2}$ graded and the sign factors in the commutation rules are given by projection onto the second factor of the grading. We assume that there exists an atlas consisting of affine charts in which the coordinate functions have both well defined  integer \emph{weight} and Grassmann \emph{parity} represented by an element of $\mathbb{Z}_2$. Furthermore, geometric objects on graded manifolds, such as vector fields, also carry weight.  The graded structure is conveniently encoded in a weight vector field, which via its Lie derivative counts the weight of homogeneous objects. Changes of local coordinates respect the weight and Grassmann parity. Similarly, morphisms in the category of graded manifold are similarly required to respect the weight and Grassmann parity.

An important class of graded manifold are those that carry  non-negative integer weights. That is the structure sheaf is $\mathbb{N}\times\mathbb{Z}_{2}$ graded.   For clarity and with our applications in mind we will further restrict to the strictly commutative case. That is, the underlying structure will be that of a genuine manifold and not a supermanifold. However, this restriction is inessential for the geometric constructions presented in this section.

We will require that  the weight vector fields on  non-negatively  graded manifolds to be \emph{h-complete} \cite{Grabowski2012}.   To explain this concept properly, let us  denote the algebra of homogenous in weight functions on $F$ as $\mathcal{A}(F) = \bigoplus_{i \in \mathbb{R}} \mathcal{A}^{i}(F)$, where $\mathcal{A}^{i}(F)$ is the vector space of  homogeneous functions $f$ of weight $i$. Note that a priori the weight of a given homogeneous function need not be integer valued, but rather the weight can be any real number, including irrational and negative. Generically the weight vector field induces a one-parameter group of automorphisms of the algebra of homogeneous functions as $\Phi_{\lambda}(f) = e^{\lambda w}f$ where $f \in \mathcal{A}^{w}(F)$. One can also think in terms of an action of the multiplicative group $(\mathbb{R}_{>0 }, \cdot)$ viz $h_{t}(f) = \Phi_{\ln(t)}(f) = t^{w}f$.

The h-completeness of the weight vector field means that the action $h_{t}$ can be extended to a smooth action of the moniod $(\mathbb{R}, \cdot)$ on $F$. In particular, this means that only non-negative integer weights are allowed; $\mathcal{A}(F) = \bigoplus_{i \in \mathbb{N}}\mathcal{A}^{i}(F)$.  This algebra, which is dense on the algebra of all smooth function on $F$ with respect to any reasonable $C^{\infty}$-topology, is referred to as the\emph{ algebra of polynomial functions} on $F$.  Such actions are known as \emph{homogeneity structures} \cite{Grabowski2012}. This means that for $t \neq 0$ the action $h_{t}$ is a diffeomorphism of $F$ and when $t=0$ it is a smooth surjection $\tau=h_0$ onto $F_{0}=M$, with the fibres being diffeomorphic to $\mathbb{R}^{N}$.  What we obtain are particular \emph{polynomial bundles} $\tau:F\to M$, i.e. fibrations which locally look like $U\times\R^N$ and the change of coordinates (for a certain choice of an atlas) are polynomial in $\R^N$. For this reason graded manifolds with non-negative weights and h-complete weight vector fields are also known as \emph{graded bundles} \cite{Grabowski2012}. The  h-completeness condition  implies that  graded bundles are determined by the algebra of homogeneous functions on them.  Importantly, each $\mathcal{A}^{i}(F)$ is a locally free and finitely generated $C^{\infty}(M)$-module, meaning that each $\mathcal{A}^{i}(F)$ is the module of sections of some vector bundle over $M$. Canonical examples of graded bundles are, for instance, vector bundles, $n$-tuple vector bundles, higher tangent bundles $\sT^kM$, and multivector bundles $\we^n\sT E$ of vector bundles $\tau:E\to M$ with respect to the projection $\we^n\sT\tau:\we^n\sT E\to \we^n\sT M$ (see \cite{Grabowska2014}).

 \begin{remark}
 While it is possible to consider manifolds with a gradation  that does not automatically lead to h-complete weight vector fields, such manifolds are far less rigid in their structure and can exhibit some pathological behavior from our point of view. We should also remark that Roytenberg,  \v{S}evera   and Voronov assume that on a graded supermanifold  the Grassmann even coordinates of non-zero weight are ``cylindrical"  \cite[Definition 4.1]{Voronov:2001qf} from the very beginning. This, together with the fact that functions of odd coordinates are automatically polynomial, means that  the weight vector fields on non-negatively graded supermanifolds or N-manifolds  are h-complete.
 \end{remark}

 \begin{example}
Consider the vector space $\R$ with the standard Euler vector field $\nabla=x\pa_x$. The open submanifold
$F=\{x\in\R:x>0\}$ with the restriction of this Euler vector field $\nabla_F={x\pa_x}_{|F}$ as the weight vector field is a graded manifold. The weight vector field is in this case complete but not h-complete, as the action $h_t(x)=tx$ has no limit at $t\to 0$ and hence cannot be extended to the action of $\R$. What is more, there are homogenous functions with non-integer degrees: $f_a(x)=x^a$ is of degree $a$ for any $a\in\R$. The homogeneous functions $f_a$ are not polynomials in the homogeneous coordinate $x$ if $a \notin \mathbb{N}$.
\end{example}

The previous example shows that it is quite possible to have weight vector fields that are complete, but not h-complete; the above graded structure is not fully encoded in an action of $(\mathbb{R}, \cdot)$. Moreover, without h-completeness  homogeneous functions are not necessarily polynomial in homogeneous coordinates. From our perspective such graded manifolds, although well defined, have a far less rigid structure.

\begin{example}
Now consider instead the graded manifold $F'=\{x\in\R:|x|<1\}$ with the corresponding restriction $\nabla_{F'}$ of
$\nabla$ as the weight vector field, then $\nabla_{F'}$ is not even complete as the corresponding homotheties $h_t$, though defined for all $|t|\le 1$, thus also for 0, cannot be extended for $|t|>1$. On the other hand, all homogeneous functions have integer non-negative degrees and are
polynomials in the homogeneous coordinate $x$. The algebra $\mathcal{A}(F')$ is the algebra of polynomials in $x$ exactly like $\mathcal{A}(\R)$ although $F'$ and $\R$ are not isomorphic as graded manifolds, as there are no diffeomorphisms between the interval $(-1,1)$ and $\R$ respecting weights.
\end{example}

The above example shows that it is  reasonable to understand a general `graded manifold' as a manifold with a consistently chosen atlas of homogeneous coordinates, without insisting on the weight vector field being complete, let alone h-complete. Clearly, we cannot understand such graded structure as being associated with a smooth action of the multiplicative monoid of reals.

\medskip
A fundamental result is that any smooth action of the multiplicative monoid  $(\R,\cdot)$ on a manifold leads to a $\mathbb{N}$-gradation of that manifold. A little more carefully, the category of graded bundles is  equivalent to the category of $(\mathbb{R}, \cdot)$-manifolds and equivariant maps. A canonical construction of this correspondence goes as follows.
Take $ p \in F$ and consider $t \mapsto h_{t}(p)$ as a smooth curve $\zg_h^p$ in $F$. This curve meets  $M=h_0(F)$ for $t=0$ and is constant for $p\in M$.

\begin{theorem}[\cite{Grabowski2012}]
For any homogeneity structure $h:\R\ti F\to F$, the subset $M=h_0(F)$ is a submanifold of $F$
and there is $k\in\mathbb{N}$ such that
the map
\be\label{Phi}\Phi_h^k:F\to \sT^k F_{|M}\,,\quad \Phi^k_h(p)=\sj^k_0(\zg_h^p)\,,
\ee
is an equivariant (with respect to the monoid actions of $(\R,\cdot)$ on $F$ and $\sT^k F$) embedding of $F$ onto a graded submanifold of the graded bundle $\sT^k F_{|M}$. In particular, there is an atlas on $F$ consisting of homogeneous function.
\end{theorem}

Any $k$ described by the above theorem we call a \emph{degree} of the homogeneity structure $h$. We stress that  a graded bundle is not just a manifold with consistently defined homogeneous local coordinates. The additional condition is that the weight vector field encoding the graded structure is  h-complete; that is the associated one-parameter group of diffeomorphisms  can be extended to an action of the monoid $(\mathbb{R}, \cdot)$. Importantly, we can use the homogeneity structure to neatly encode the graded structure and phrase much of the theory accordingly.

One can pick an affine (maximal) atlas of $F$ consisting of charts for which we  have homogeneous local coordinates $(x^{A}, y_{w}^{a})$, where $\w(x^{A}) =0$ and  $\w(y_{w}^{a}) = w$ with $1\leq w \leq k$, for some $k \in \mathbb{N}$ known as the \emph{degree}. That is we will group all the coordinates with non-zero weight together.  Here $a$ should be considered as a ``generalised index" running over all the possible weights. The label $w$ in this respect is somewhat redundant, but it will come in very useful when checking the weight of various expressions.  The local changes of coordinates  are of the form

\begin{eqnarray}\label{eqn:translaws}
x^{A'} &=& x^{A'}(x),\\
\nonumber y^{a'}_{w} &=& y^{b}_{w} T_{b}^{\:\: a'}(x) + \sum_{\stackrel{1<n  }{w_{1} + \cdots + w_{n} = w}} \frac{1}{n!}y^{b_{1}}_{w_{1}} \cdots y^{b_{n}}_{w_{n}}T_{b_{n} \cdots b_{1}}^{\:\:\: \:\:\:\:\:a'}(x),
\end{eqnarray}
where $T_{b}^{\:\: a'}$ are invertible and the $T_{b_{n} \cdots b_{1}}^{\:\:\: \:\:\:\:\:a'}$ are symmetric in lower indices.

A graded bundle  of degree $k$  admits a sequence of  polynomial fibrations, where a point of $F_{l}$ is a class of the points of $F$ described  in an affine coordinate system by the coordinates of weight $\leq l$, with the obvious tower of  surjections

\begin{equation}\label{eqn:fibrations}
F=F_{k} \stackrel{\tau^{k}}{\longrightarrow} F_{k-1} \stackrel{\tau^{k-1}}{\longrightarrow}   \cdots \stackrel{\tau^{3}}{\longrightarrow} F_{2} \stackrel{\tau^{2}}{\longrightarrow}F_{1} \stackrel{\tau^{1}}{\longrightarrow} F_{0} = M,
\end{equation}

\noindent where the coordinates on $M$ have zero weight. Note that  $F_{1} \rightarrow M$ is a linear fibration and the other fibrations $F_{l} \rightarrow F_{l-1}$ are affine fibrations in the sense that the changes of local coordinates for the fibres are linear plus and additional additive terms of appropriate weight (cf. \ref{eqn:translaws}).  The model fibres here are $\mathbb{R}^{N}$ (cf.  \cite{Grabowski2012}).  We will also on occasion employ the projections $\tau^{i}_{j} : F_{i} \rightarrow F_{j}$, which are defined as the appropriate composition of the above projections. We will also use on occasion $\tau := \tau^{k}_{0} : F_{k} \rightarrow M$.\smallskip

There is also a ``dual"  sequence of submanifolds and their inclusions

\begin{equation}\label{eqn:submanifolds}
M := F_{0}= F^{[k]}  \hookrightarrow F^{[k-1]} \hookrightarrow \cdots \hookrightarrow F^{[0]} = F_{k},
\end{equation}

\noindent where we define, locally but correctly,

\begin{equation*}
F^{[i]} := \left\{\left. p \in F_{k} \right| y_{w}^{a} = 0 \:\:\: \textnormal{if} \:\: w \leq i   \right\}.
\end{equation*}

\noindent If we define $\mathcal{A}_{l}(F)$ to be the subalgebra of $\mathcal{A}(F)$ locally generated by functions of weight $\leq l$, then corresponding to the sequence of submanifolds and their inclusions is the filtration of algebras $C^{\infty}(M) = \mathcal{A}_{0}(F) \subset \mathcal{A}_{1}(F) \subset \cdots \subset \mathcal{A}_{k}(F) = \mathcal{A}(F)$.

\noindent It will also be useful to consider the submanifold $\bar{F}_{l}:= F_{l}^{[l-1]}$ which is in fact linearly fibred   over $M$, with the linear coordinates carrying weight $l$.  The module of sections of $\bar{F}_{l}$ is identified with the $C^{\infty}(M)$-module $\mathcal{A}^{l}(F)$. In words, ``\emph{you project higher to lower, but set to zero lower to higher}".

\begin{example}
The degree two case is sufficient to illustrate the general situation. On a graded manifold of degree $2$, which we denote as $F_{2}$, we can find homogeneous local coordinates $(x^{A}, y^{a}, z^{i})$ of weight $0,1,2$ respectively. The admissible changes of local coordinates are of the form;

\begin{equation}
x^{A'} = x^{A'}(x), \hspace{10pt} y^{a'} = y^{b}T_{b}^{\:\: a'}(x), \hspace{10pt}
z^{i'} = z^{j}T_{j}^{\:\: i'}(x) + \frac{1}{2!}y^{a}y^{b}T_{ba}^{\:\:\: i'}(x).
\end{equation}

From these transformation laws it is clear that $F_{2} \rightarrow F_{1}$ is an affine fibration while $F_{1}\rightarrow F_{0} =M$ is a linear  fibration. Also it is easy to see that $F^{[1]}$ admits local coordinates $(x^{A}, z^{i})$. Due to the admissible changes of local coordinates   $F^{[1]} \rightarrow M$ is again a linear fibration. Note that trying to construct a submanifold of $F_{2}$ by setting $z^{i} =0$ is not going to work due to the affine nature of the transformation laws. Also note that for this low degree example we have $\bar{F}_{2} = F^{[1]}_{2} = F^{[1]}$.
\end{example}

From the transformation laws of the local coordinates on $F$ it is clear that for any graded bundle of degree $k$ we have a filtration of integrable vertical distributions associated with the series of projections (\ref{eqn:fibrations})
 \begin{equation}
 \mathcal{V}_{k} \subset \mathcal{V}_{k-1} \subset \cdots \subset \mathcal{V}_{1} \subset \sT F_{k},
\end{equation}

\noindent which admit local bases $\mathcal{V}_{1} = \Span \{ \frac{\partial}{\partial y_{1}} , \frac{\partial}{\partial y_{2}} , \cdots \frac{\partial}{\partial y_{k}}\}$, $\mathcal{V}_{2}= \Span\{ \frac{\partial}{\partial y_{2}} , \frac{\partial}{\partial y_{3}} , \cdots \frac{\partial}{\partial y_{k}}\}$, $\cdots $, $\mathcal{V}_{k}= \Span \{ \frac{\partial}{\partial y_{k}}\}$, where we  have suppressed the indices for clarity. For example, the weight vector field, given in homogeneous  local coordinates by

\begin{equation}
\Delta_{F} = \sum_{w=1}^{k}  w \: y^{a}_{w} \frac{\partial}{\partial y^{a}_{w}},
\end{equation}

\noindent belongs to $\mathcal{V}_{1}$, that is, it is vertical with respect to the projection $\tau : F_{k} \rightarrow M$.

Morphisms  between graded bundles are necessarily polynomial in the non-zero weight coordinates and respect the weight. Again, one can apply the symbolic notation to get at the structure of morphisms between graded bundles. Morphisms of graded bundles necessarily  intertwine the respective weight vector fields \cite{Grabowski2012,Voronov:2001qf}, or equivalently the respective homogeneity structures \cite{Grabowski2012}.

\begin{example}
The standard examples of graded bundles include;
 \begin{enumerate}
 \item Any manifold if we assign weight zero to all the coordinates.
 \item Total spaces of vector bundles  $\tau: E \rightarrow M$ if we assign any positive weight to the linear coordinates and zero to the base coordinates. The natural choice is to assign weight one to  the fibre coordinates.
 \item The  $k$-th order tangent bundle $\sT^{k}M$ of a manifold $M$ naturally has the structure of a graded bundle of degree $k$.
 \item The tangent bundle of $k^{1}$-velocities $\sT^{\times k}M = \sT M \times_{M}\sT M \times_{M}\cdots \times_{M}\sT M$ (k-times) is graded bundle of degree $k$ by assigning weight $i$ $(1\leq i \leq k)$ to  the linear coordinates of the $i$-th factor.
 \item A little more generally  a \emph{split graded bundle} is of the form $F_{k} = E_{1}\times_{M}E_{2} \times_{M} \cdots \times_{M}E_{k}$ where each $E_{i}$ is the total space of some vector bundle and we assign weight $i$ to the linear coordinates  of $E_i$. Clearly graded bundles of degree $1$ are always split as they correspond to vector bundles.
 \item Multivector bundles $\we^n\sT E$ over vector bundles $\tau:E\to M$ (or, more generally, graded bundles) with respect to the projection $\we^n\sT\tau:\we^n\sT E\to \we^n\sT M$ (see \cite{Grabowska2014}).
 \item Inductively: tangent (more generally: higher tangent) and cotangent bundles, $\sT^rF$ and $\sT^*F$ of graded bundles $F$ (see \cite{Grabowski2009,Grabowski2013} and the next section).
\end{enumerate}
\end{example}

\begin{remark}
Note that we have a Gaw\c{e}dzki--Batchelor-like theorem here stating that every graded bundle is \emph{non-canonically isomorphic} to a split graded bundle. The split form of a graded bundle \emph{is} canonical and given by  $F_{k} \simeq \bar{F}_{1}\times_{M}\bar{F}_{2}\times_{M} \cdots \times_{M} \bar{F}_{k} $, but the isomorphism itself  is non-canonical. While this has been folklore for a number of years in the super-context, the first formal proof can be found in \cite{Bonavolonta:2013}.  Recall that every (real) supermanifold is isomorphic to a supermanifold  of the form $\Pi E$ where $E$ is an ordinary smooth vector bundle. However, the category of supermanifolds is not equivalent to that of vector bundles: a morphism of vector bundles is equivalent to a morphism of supermanifolds that is strictly homogeneous in weight, while a general morphism of supermanifolds  need only be parity preserving. A similar situation arises in the context of graded bundles. In particular the category of graded bundles is \emph{not} equivalent to that of vector bundles as in general we allow  morphisms that are not strictly  linear as long as they preserve the weight,  we know that they are generally polynomial. We will not make any explicit use of a splitting in this work, other than the partial splitting described below.
\end{remark}

\begin{proposition}\label{prop:partialsplitting}
 Any graded bundle $F_{k}$ of degree $k>1$ can always be non-canonically partially split as $F_{k} \simeq F_{k-1} \times_{M}\bar{F}_{k}$.
\end{proposition}

\begin{proof}
For simplicity we will assume that $F_{0} = M$ is connected. Denote with $\mathcal{A}^{\langle k\rangle}(F)$ the vector bundle of polynomials of degree $k$ in variables of degree smaller that $k$ (it is a locally free sheaf of $C^{\infty}(M)$-modules). In other words, if $\operatorname{Alg}(\mathcal{A}_{k-1}(F))$ is the algebra generated by functions of degrees $\le k-1$, $\mathcal{A}_{k-1}(F) = \mathcal{A}^{0}(F)\oplus\mathcal{A}^{1}(F) \oplus \cdots \oplus\mathcal{A}^{k-1}(F)$, then
$$\mathcal{A}^{\langle k\rangle}(F)=\operatorname{Alg}\left(\mathcal{A}_{k-1}(F)\right) \cap \mathcal{A}^{k}(F)$$
and
\begin{equation*}
\mathcal{A}^{k}(F)\setminus\mathcal{A}^{\langle k\rangle}(F) \simeq \bar{F}_{k}^*\,.
\end{equation*}
Thus, we have the short exact sequence of $C^{\infty}(M)$-modules
\begin{equation*}
0 \rightarrow \mathcal{A}^{\langle k\rangle}(F) \cap \mathcal{A}^{k}(F) \longrightarrow \mathcal{A}^{k}(F) \longrightarrow \bar{F}_{k}^* \rightarrow 0,
\end{equation*}
which is non-canonically split as $\Gamma(\bar{F}_{k}^*)$ is a projective module. Then with a chosen splitting it follows that
\begin{equation*}
\mathcal{A}(F) = A_{k-1}(F) \oplus_{C^{\infty}(M)}\bar{F}_{k}^*.
\end{equation*}
Then with a  choice of splitting we see that $F_{k}$ is diffeomorphic to $F_{k-1} \times_{M}\bar{F}_{k}$.
\end{proof}

\begin{remark}
The above argument can then be iterated to show that any graded bundle can non-canonically be split. The downside to the Gaw\c{e}dzki--Bachelor-like theorem for graded bundles is that in general we have no preferred choice of splitting, we only know that splittings always exist.
\end{remark}

\begin{example}
It is well known that the second order tangent bundle splits as  $\sT^{2}M \simeq \sT M \times_{M}\sT M $, but of course non-canonically. A little more generally we have $\sT^{k}M \simeq \sT^{k-1}M \times_{M} \sT M$, with the obvious assignment of weights. Iterating this we see that the k-th order tangent bundle is diffeomorphic to the tangent bundle of $k^{1}$-velocities, $\sT^{k}M \simeq \sT^{\times k}M$, albeit non-canonically.
\end{example}

\begin{example}
Similarly to the pervious example  $\sT (\sT^{*}M) \simeq \sT^{*} (\sT M) \simeq \sT M\times_{M} \sT^{*}M \times_{M} \sT^{*}M$ where we consider the total weight of the double vector bundles in question. Again this splitting  is  non-canonical. See the next subsection for details on double graded bundles.
\end{example}

\begin{example}
Sometimes a splitting can be canonical. It is well known that for a vector bundle  $\pi:E\rightarrow M$, the associated vertical bundle canonically splits as $VE \simeq E \times_{M} E \subset \sT E$. In the context of graded bundles the homogeneous local coordinates on $VE$ are $(x, y , \dot{y})$, which can naturally be assigned the weights of $0, 1$ and $2$ respectively.  Thus we have an example of a canonically split graded bundle of degree $2$. In fact this example generalises to $V^{k}E := \ker(\sT^{k}\pi ) \simeq E \times_{M} E \times_{M} \cdots \times_{M} E \subset \sT^{k} E$, where the fibred product is repeated $k$-times. Again the weight is assigned as $i$ to the fibre coordinates of the $i$-th factor and so we have a canonically split graded bundle of degree $k$.
\end{example}

\subsection{Double and $n$-tuple graded bundles}
To recap, a graded bundle is really a pair $(F, \Delta)$ where $F$ is a smooth manifold for which the structure sheaf carries an $\mathbb{N}$-grading and $\Delta$ is the associated h-complete weight vector field. For example, if the weight is constrained to be either zero or one, then the weight vector field is precisely a  vector bundle structure on $F$ and will be generally referred  to as an \emph{Euler vector field}. The notion of a double vector bundle \cite{Pradines:1974} (or a higher $n$-tuple vector bundle)   is conceptually  clear in the graded language in terms of mutually commuting weight vector fields; see  \cite{Grabowski2009,Grabowski2012}. This leads to the higher analogues known as \emph{$n$-tuple graded bundles}.

\begin{definition}
An  \emph{$n$-tuple graded bundle} is a manifold $\mathcal{M}$  for which  its structure sheaf carries an $\mathbb{N}^{n}$-grading such that all the weight vector fields are h-complete and pairwise commuting.  In particular a \emph{double graded bundle} consists of a manifold and a pair of mutually commuting weight vector fields. If all the weights are either zero or one then we speak of an \emph{$n$-tuple vector bundle}.
\end{definition}

\begin{remark}
Equivalently, one can define an $n$-tuple graded bundle in terms of pairwise commuting homogeneity structures. Depending questions posed, using homogeneity structures may offer direct answers that are far less tractable than using local coordinates or weight vector fields.
\end{remark}

 Let $\mathcal{M}$ be an $n$-tuple graded bundle and $(\Delta^{1}_{\mathcal{M}} , \cdots , \Delta^{n}_{\mathcal{M}})$ be the collection of pairwise commuting weight vector fields.The local triviality of $n$-tuple graded bundles was established in \cite{Grabowski2012}. This means that we can always equip an $n$-tuple graded bundle with  an atlas such that the charts consist of coordinates that are simultaneously homogeneous with respect to the  weights associated with each weight vector field. Thus, we can always equip an $n$-tuple graded bundle with homogeneous local coordinates $(y)$ such that $\underline{\w}(y) = (\w_{1}(y) , \cdots \w_{n}(y) ) \in \mathbb{N}^{n}$. The changes of local coordinates must respect the weights.

Each $\Delta^{s}_{\mathcal{M}}$  $(1 \leq s \leq n)$ defines a submanifold $\mathcal{M}_{s} \subset \mathcal{M}$ for which $\Delta^{s}_{\mathcal{M}}=0$. Moreover, as the weight vector fields are h-complete, we have a bundle structure, that is a surjective submersion $h_{0}^{s}:\mathcal{M} \rightarrow \mathcal{M}_{s}$ defined by the homogeneity structures associated with  each $\Delta^{s}_{\mathcal{M}}$.  As the homogeneity structures (or equivalently the weight vector fields) all pairwise commute, we have a whole diagram of fibrations  $h_{0}^{i_{1}}\circ h_{0}^{i_{2}}\circ \cdots \circ h_{0}^{i_{n}}: \mathcal{M} \rightarrow \mathcal{M}_{i_{1}}\cap \mathcal{M}_{i_{2}} \cap \cdots \cap \mathcal{M}_{i_{n}}$, where the fibres are homogeneity spaces (cf. \cite{Grabowski2012}).

 \begin{example}All $n$-tuple vector bundles, e.g. double vector bundles, are examples of $n$-tuple graded bundles. Here each of the weights that make up the multi-weight are restricted to be either zero or one. Thus, $n$-tuple vector bundles are encoded in a collection of n pairwise commuting Euler vector fields.
 \end{example}

\begin{example}
Consider an arbitrary graded bundle $F_{k}$ of degree $k$. Let us employ  homogeneous coordinates $(x^{A}, y_{w}^{a})$ with $1 \leq w \leq k$. The weight vector field is as before given by
\begin{equation*}
\Delta_{F} := \sum_{w} w y_{w}^{a}\frac{\partial}{\partial y_{w}^{a}}.
\end{equation*}
The tangent bundle $\sT F_{k}$ naturally has the structure of a double graded bundle. To see this we employ homogeneous coordinates  with respect to the bi-weight which consists of the natural lifted weight and the weight due the vector bundle structure of the tangent bundle \cite{Grabowski:1995,Yano:1973};

\begin{equation*}
 (\underbrace{x^{A}}_{(0,0)}, \hspace{5pt} \underbrace{y^{a}_{w}}_{(w,0)},\hspace{5pt} \underbrace{\dot{x}_{1}^{B}}_{(0,1)}, \hspace{5pt}  \underbrace{\dot{y}^{b}_{w+1}}_{(w,1)}).
 \end{equation*}
 In other words, we have the first weight vector field being simply the \emph{tangent lift} \cite{Grabowski:1995,Yano:1973} of the weight vector field in $F_{k}$ and the second being the natural Euler  vector field on a tangent bundle;
 \begin{eqnarray}
 \Delta^{1}_{\sT F} &=& \D_{\sT}\Delta_{F} = \sum_{w} w y_{w}^{a}\frac{\partial}{\partial y_{w}^{a}} + \sum_{w} w \dot{y}_{w+1}^{a}\frac{\partial}{\partial \dot{y}_{w+1}^{a}},\\
 \nonumber \Delta^{2}_{\sT F} &=& \dot{x}_{1}^{A}\frac{\partial}{\partial \dot{x}_{1}^{A}} + \sum_{w} \dot{y}_{w+1}^{a}\frac{\partial}{\partial \dot{y}_{w+1}^{a}}.
 \end{eqnarray}
 \begin{tabular}{p{10cm} p{5cm}}
\noindent It is not hard to see that the zeros of these weight vector fields describe the following  diagram of graded bundles and their morphisms:

&
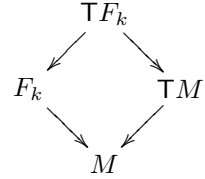
\vspace{-20pt}
\begin{center}
\leavevmode
\begin{xy}
(10,50)*+{\sT F_{k}}="a";%
(0,40)*+{F_{k}}="b"; (20,40)*+{\sT M}="c";%
(10,30)*+{ M}="d";%
{\ar "a";"b"};%
{\ar "a";"c"};%
{\ar "b";"d"};%
{\ar "c";"d"};%
\end{xy}
\end{center}
\end{tabular}\\
 This can be slightly generalized by considering multivector bundles $\we^r\sT F_k$. They are vector bundles over $F_k$, but generally higher graded bundles with respect to the projection
$$\we^r\sT\tau:\we^r\sT^rF_k\to\we^r\sT M\,.$$
We will not consider this case in detail here just referring to the paper \cite{Grabowska2014}.

Within this example we see that the tangent bundle $\sT E$ of a vector bundle $E$ is a double vector bundle, as is well know. What is less known is the fact that $\we^r\sT E$ is no longer a double vector bundle but a double bundle for which one structure is a vector bundle and the second
is a graded bundle (that is vector only if $r\le 1$).
\end{example}

One can construct new weight vector fields on $\mathcal{M}$,  by taking positive linear combinations of the weight vector fields; eg. $\Delta_{\mathcal{M}} = a_{1} \: \Delta^{1}_{\mathcal{M}} + a_{2} \: \Delta^{2}_{\mathcal{M}} + \cdots a_{n}  \: \Delta^{n}_{\mathcal{M}}$ where $a_{I} \in \mathbb{N}$. However, note that in general positive linear combinations of Euler vector fields are not Euler vector fields; i.e. $n$-tuple vector bundles are not really just vector bundles in a different guise. Also  note that sometimes it is possible to take linear combinations with some negative coefficients and still produce a weight vector field. We shall see an example of this in the next section.

\begin{example}
Similarly to the case of the tangent bundle, the cotangent bundle $\sT^{*}F_{k}$ of a graded bundle comes naturally with the structure of a double graded bundle. However,  simply using the naturally induced weight would mean that the ``momenta" will have a negative component of their bi-degree and so the associated weight vector field cannot be h-complete; we would not remain in the category of graded bundles. Instead one needs to consider a \emph{phase lift} of the weight vector field on $F_{k}$  (cf. \cite{Grabowski2013}). The homogeneous coordinates naturally induced on $\sT^{*}F_{k}$ are
\begin{equation*}
 (\underbrace{x^{A}}_{(0,0)}, \hspace{5pt} \underbrace{y^{a}_{w}}_{(w,0)},\hspace{5pt} \underbrace{p^{1}_{B}}_{(0,1)}, \hspace{5pt}  \underbrace{p_{b}^{1-w}}_{(-w,1)}),
 \end{equation*}
where $1\leq w \leq k$. The associated weight vector fields are
\begin{eqnarray}
 \bar{\Delta}^{1}_{\sT^{*}F} &=& \D_{\sT^{*}}\Delta = \sum_{w} w y_{w}^{a}\frac{\partial}{\partial y_{w}^{a}} -\sum_{w} w p_{a}^{1-w}\frac{\partial}{\partial p_{a}^{1-w}},\\
 \nonumber \bar{\Delta}^{2}_{\sT^{*}F} &=& p^{1}_{A}\frac{\partial}{\partial p^{1}_{A}} + \sum_{w} p^{1-w}_{a}\frac{\partial}{\partial p^{1-w}_{a}}.
 \end{eqnarray}
 The \emph{k-phase lift} of the weight vector field $\Delta_{F}$ essentially  produces a shift in the bi-weight to ensure that everything is non-negative. It amounts  to a redefinition of the first weight vector field viz $\Delta^{1}_{\sT^{*} F} =  \bar{\Delta}^{1}_{\sT^{*}F} +  k\: \bar{\Delta}^{2}_{\sT^{*}F}$. With a slight relabeling of our homogeneous coordinates we have
 \begin{equation*}
 (\underbrace{x^{A}}_{(0,0)}, \hspace{5pt} \underbrace{y^{a}_{w}}_{(w,0)},\hspace{5pt} \underbrace{\pi^{k+1}_{B}}_{(k,1)}, \hspace{5pt}  \underbrace{\pi_{b}^{k-w+1}}_{(k-w,1)}),
 \end{equation*}
 and the weight vector fields are now
 \begin{eqnarray}
 \Delta^{1}_{\sT^{*}F} &=&  \sum_{w} w y_{w}^{a}\frac{\partial}{\partial y_{w}^{a}} + k \pi_{A}^{k+1} \frac{\partial }{\partial \pi_{A}^{k+1}} + \sum_{w} (k-w) \pi_{a}^{k-w+1}\frac{\partial}{\partial \pi_{a}^{k-w+1}},\\
 \nonumber \Delta^{2}_{\sT^{*}F} &=& \pi_{A}^{k+1}\frac{\partial}{\partial \pi^{k+1}_{A}} + \sum_{w} \pi^{k-w+1}_{a}\frac{\partial}{\partial \pi^{k-w+1}_{a}}.
 \end{eqnarray}
\begin{tabular}{p{10cm} p{5cm}}
\noindent It is not hard to see that the zeros of these weight vector fields describe the following diagram of graded bundles and their morphisms:

&
\vspace{-20pt}
\begin{center}
\leavevmode
\begin{xy}
(10,50)*+{\sT^{\ast} F_{k}}="a";%
(0,40)*+{F_{k}}="b"; (20,40)*+{\bar{F}_{k}}="c";%
(10,30)*+{ M}="d";%
{\ar "a";"b"};%
{\ar "a";"c"};%
{\ar "b";"d"};%
{\ar "c";"d"};%
\end{xy}
\end{center}
\end{tabular}\\

Similarly as above we can obtain a double graded bundle $\we^r\sT^* F_k$ which is a vector bundle over $F_k$, but generally a higher graded bundle with respect to the canonical projection $\we^r\sT^*F_k\to\we^r\bar{F}^{*}_{k}$. In particular, we see that the cotangent bundle $\sT^* E$ of a vector bundle $E$ is a double vector bundle with the second projection onto $E^*$, as is well-know.
\end{example}

\begin{warning}
Throughout this paper we will equip the cotangent bundles of graded bundle with  the k-phase lifted weight vector field, to ensure do not leave the category of $n$-tuple graded bundles.
\end{warning}

\begin{example}
Via iterating the examples of the tangent and cotangent bundles of a graded bundle we see that $\sT \sT F_{k}$, $\sT \sT^{*} F_{k},\sT^{*} \sT F_{k}$ and $\sT^{*}\sT^{*}M$ are triple graded bundles. The obvious similar statements hold for higher iterations of $\sT$ and $\sT^{*}$.
\end{example}

One can pass from an $n$-tuple graded bundle to an (n-r)-tuple graded bundle via the process of taking linear combinations of weight vector fields.  In particular one can construct the \emph{total weight vector field} as the sum of all the weight vector fields on an $n$-tuple graded bundle. However, note that the passage from an $n$-tuple graded bundle to a graded bundle via collapsing the weights is far from unique.

\begin{example}
The tangent bundle $\sT F_{k}$ of a graded bundle of degree $k$ can be considered as a graded bundle of degree $k+1$ via the total weight which is described by the total weight vector field
\begin{equation*}
\Delta_{\sT F_{k}} := \sum_{w} w y_{w}^{a}\frac{\partial}{\partial y_{w}^{a}} + \dot{x}_{1}^{A}\frac{\partial}{\partial \dot{x}_{1}^{A}} + \sum_{w} (w+1)\dot{y}^{a}_{w+1} \frac{\partial}{\partial \dot{y}_{w+1}^{a}}.
\end{equation*}
\end{example}

\begin{example}
The cotangent bundle $\sT^{*} F_{k}$ of a graded bundle of degree $k$ can be considered as a graded bundle of degree $k+1$ via the total weight which is described by the total weight vector field
\begin{equation*}
 \Delta_{\sT^{*}F_{k}} := \sum_{w} w y_{w}^{a}\frac{\partial}{\partial y_{w}^{a}} + (k+1) \pi_{A}^{k+1} \frac{\partial }{\partial \pi_{A}^{k+1}} + \sum_{w} (k-w+1) \pi_{a}^{k-w+1}\frac{\partial}{\partial \pi_{a}^{k-w+1}}.
 \end{equation*}
 It is this weight that Roytenberg \cite{Roytenberg:2001} employed for the case of a vector bundle  $F_{k} = E$ to obtain canonical examples of symplectic N-manifolds of degree 2.
\end{example}

\begin{remark}
We will choose to label  homogeneous local  coordinates of $n$-tuple graded bundles by the total weight, as we have already done so in the previous examples. This choice will be convenient for checking the consistency of various expression involving local coordinates.
\end{remark}

A general $n$-tuple graded bundle admits many locally trivial fibrations compatible with the graded structure. The theory is already richer for double graded bundles than the classical double vector bundles. These fibrations can be seen as a choice of a  graded bundle structure via taking a suitable linear combination of the weight vector fields.

To set some notation, if $\Delta$ is a weight vector field on $\mathcal{M}$, then we denote by $\mathcal{M}[\Delta \leq l]$ the base manifold of the locally trivial fibration defined by taking the weight $>l$ coordinates with respect to this complimentary  weight to be the fibre coordinates. We have a natural projection that we will denote as
 \begin{equation*}
\textnormal{p}^{\mathcal{M}}_{[\Delta \leq l]}: \mathcal{M} \rightarrow \mathcal{M}[\Delta \leq l].
\end{equation*}
The following is a simple but very useful observation.
\begin{proposition}\label{ntproj}
Let $\mathcal{M}$ be an $n$-tuple graded bundle and $(\Delta^{1}_{\mathcal{M}} , \cdots , \Delta^{n}_{\mathcal{M}})$ be the corresponding pairwise commuting weight vector fields.
Then the weight vector fields $(\Delta^{1}_{\mathcal{M}} , \cdots , \Delta^{n}_{\mathcal{M}})$ are $\textnormal{p}^{\mathcal{M}}_{[\Delta^1 \leq l]}$-projectable and their projections define a structure of an $n$-tuple graded bundle on $\mathcal{M}[\zD^1\le l]$. With respect to this structure, $\textnormal{p}^{\mathcal{M}}_{[\Delta \leq l]}: \mathcal{M} \rightarrow \mathcal{M}[\Delta^1 \leq l]$ is a morphism of $n$-tuple graded bundles.
\end{proposition}
\begin{proof}
It follows from the fact that the $(\R,\cdot)$-actions corresponding to the projections can be viewed as the restrictions of the original actions to the subalgebra generated by $\zD^1$-homogeneous functions of degree $\le l$. Such restrictions exist due to commutativity of the actions.
\end{proof}

\begin{example}
By employing the total weight we have $\sT E[\Delta_{\sT E} \leq 1]  \simeq E \times_{M} TM$ for any vector bundle $E$ considered as a graded bundle.
\end{example}

\begin{example}
Similarly, by employing the total weight we  have $\sT^{*} E[\Delta_{\sT^{*} E} \leq 1]  \simeq E \times_{M} E^{*}$ for any vector bundle $E$ considered as a graded bundle. In this way the \emph{Pontryagin bundle} or \emph{generalised tangent bundle} $\sT M \times_{M} \sT^{*}M$ can be constructed from $\sT \sT^{*}M \simeq \sT^{*}\sT^{*}M $ together with the corresponding projection.
\end{example}

In almost an identical way to the graded bundle case (cf. (\ref{eqn:submanifolds})), we can define certain submanifolds  of $\mathcal{M}$ by consistently setting particular coordinates to zero in a well defined manner. First, for  $\underline{l}$ and $\underline{m} \in \mathbb{N}^{n}$ we can define a partial ordering as $\underline{l} \preceq \underline{m} \Leftrightarrow \forall j \:\: l_{j} \leq m_{j}$. Then we can define

\begin{equation}
\mathcal{M}^{[\: \underline{m}\:  ]} :=  \left\{ \left. p \in \mathcal{M} \right| y = 0 \:\: \textnormal{if} \:\: \underline{\w}(y) \neq \underline{0} \: \: \textnormal{and} \:\: \underline{\w}(y) \preceq \underline{m} \right\}.
\end{equation}

As the changes of local coordinates must respect each of the weights independently and the fact that all the weights are non-negative  it is clear that  $\mathcal{M}^{[\underline{m} ]}$  are well defined $n$-tuple graded subbundles of $\mathcal{M}$.

\begin{example}\label{ex1}
It is easy to see that with respect to the standard bi-weight on $\sT F_{k}$ we have $\sT F_{k}^{[(0,1)]} = VF_{k}$, that is the vertical bundle with respect  to the projection $\tau: F_{k} \rightarrow M$. The weight vector fields on the vertical bundle are simply the appropriate restrictions of those on the tangent bundle.  Via passing to the total weight we see that $VF_{k}$ can be considered as a graded bundle of degree $k+1$. Let us remark that one can shift the first component of bi-weight to allow us to consider the vertical bundle as a double graded bundle of total degree $k$.  In terms of the weight vector fields this is simply the redefinition $\Delta_{VF}^{1} \mapsto \Delta_{VF}^{1}- \Delta_{VF}^{2}$ and $\Delta_{VF}^{2} \mapsto \Delta_{VF}^{2}$ and then considering the redefined total weight.  We will use this assignment of weights for the vertical bundle in the next section.
\end{example}

For further details of graded (super)manifolds see for example \cite{Cattaneo:2011,Grabowski2012,Grabowski2013,Roytenberg:2001,Voronov:2001qf}. For details of the classical theory of double vector bundles the reader can consult Mackenzie \cite{Mackenzie2005} and/or Konieczna \& Urba\'{n}ski \cite{Konieczna:1999}.  The original definition of a double vector bundle by Pradines \cite{Pradines:1974} has rather complicated compatibility conditions between the vector bundle structures. Grabowski \& Rotkiewicz \cite{Grabowski2009} showed how double vector bundles can be understood in terms of scalar multiplication only and also provided a proper understanding of this in terms of commuting Euler vector fields.

\section{Linearisations and linear duals}\label{sec:associating double}
\subsection{The linearisation  of a  graded bundle}
It is well know that there is a canonical embedding  $\sT^{k}M \subset \sT(\sT^{k-1}M)$ via ``holonomic vectors".  In this subsection we wish to mimic this embedding  for general graded bundles as closely as possible. To pass from a graded bundle to a double graded bundle we employ a systematic linearisation procedure. Obviously, $\sT(\sT^{k-1}M) \rightarrow \sT^{k-1}M$ is a vector bundle and it is this feature that we wish to generalise to arbitrary graded bundles.

Consider $F_{k}$ equipped with local coordinates $(x^{A}, y_{w}^{a}, z^{i}_{k})$, where the weights are assigned as $\w(x)=0$, $\w(y_{w}) = w$ $(1\leq w < k)$ and $\w(z) =k$. From now on it will be convenient to single out the highest weight coordinates as well as the zero weight coordinates. From the previous section we can construct the vertical bundle as $VF_{k} = \sT F_{k}^{[(0,1)]}$ which is really the starting place of the linearisation  procedure. Let us on $VF_{k}$ employ homogeneous local coordinates with the induced bi-weight

\begin{equation*}
 (\underbrace{x^{A}}_{(0,0)}, \hspace{5pt} \underbrace{y^{a}_{w}}_{(w,0)},\hspace{5pt}  \underbrace{z^{i}_{k}}_{(k,0)}; \hspace{5pt}  \underbrace{\dot{y}^{b}_{w}}_{(w-1,1)}, \hspace{5pt} \underbrace{\dot{z}_{k}^{j}}_{(k-1,1)}),
 \end{equation*}

\noindent where we have shifted the first component of the weight of the linear fibre coordinates so that the vertical bundle itself a graded bundle of degree $k$ by employing the total weight.

\begin{definition}
The \emph{linearisation of a graded bundle} $F_{k}$ is the double graded bundle defined as
\begin{equation*}
D(F_{k}) := VF_{k}[\Delta_{VF_k}^{1} \leq k-1]\,,
\end{equation*}
so that we have the natural projection $\textnormal{p}^{VF_{k}}_{D(F_{k})} : VF_{k} \rightarrow D(F_{k})$.
\end{definition}
Of course, the double graded bundle structure comes from that on $VF_k$ (cf. Proposition \ref{ntproj} and Example \ref{ex1}), so the first structure is of degree $k-1$ and the second is  of degree 1 (linear).

In simpler terms, in any natural homogeneous system on coordinates on the vertical bundle $VF_{k}$ one projects out the highest weight coordinates on $F_{k}$ to obtain $D(F_{k})$. One can easily check in local coordinates that doing so is well defined. Thus on $D(F_{k})$ we have local homogeneous coordinates

\begin{equation*}
 (\underbrace{x^{A}}_{(0,0)}, \hspace{5pt} \underbrace{y^{a}_{w}}_{(w,0)}; \hspace{5pt}  \underbrace{\dot{y}^{b}_{w}}_{(w-1,1)}, \hspace{5pt} \underbrace{\dot{z}_{k}^{i}}_{(k-1,1)}).
 \end{equation*}
and their associated weight vector fields

\begin{eqnarray}
 \zD^1=\Delta^{1}_{D(F_{k})}&=&\sum_{1\leq w < k} w \: y_{w}^{a} \frac{\partial}{\partial y_{w}^{a}} + \sum_{1\leq w < k} (w-1) \: \dot{y}_{w}^{a} \frac{\partial}{\partial \dot{y}_{w}^{a}} + (k-1) \: \dot{z}_{k}^{i}\frac{\partial }{\partial \dot{z}_{k}^{i}}\\
\nonumber \zD^2=\Delta^{2}_{D(F_{k})} &=& \sum_{1\leq w < k}  \dot{y}_{w}^{a} \frac{\partial}{\partial \dot{y}_{w}^{a}} +  \dot{z}_{k}^{i}\frac{\partial }{\partial \dot{z}_{k}^{i}}.
\end{eqnarray}
\noindent Note that with this assignment of the weights the linearisation of a graded bundle of degree $k$ is itself a graded bundle of degree $k$ when passing from the bi-weight to the total weight associated with
\be
\zD=\Delta_{D(F_{k})}=\Delta^{1}_{D(F_{k})}+\Delta^{2}_{D(F_{k})} =\sum_{1\leq w < k} w \: y_{w}^{a} \frac{\partial}{\partial y_{w}^{a}} + \sum_{1\leq w < k} w \: \dot{y}_{w}^{a} \frac{\partial}{\partial \dot{y}_{w}^{a}} + k \: \dot{z}_{k}^{i}\frac{\partial }{\partial \dot{z}_{k}^{i}}\,.
\ee

For completeness, the  admissible changes of local coordinates for the dotted coordinates are:

\begin{eqnarray}\label{eqn:dottedtranslaws}
\dot{y}^{a'}_{w} &=& \dot{y}_{w}^{b} T_{b}^{\:\: a'}(x) + \sum_{w_{1}+ \cdots+  w_{n} =w} \frac{1}{(n-1)!}\:\: \dot{y}^{b_{1}}_{w_{1}} y^{b_{2}}_{w_{2}} \cdots y^{b_{n}}_{w_{n}} T_{b_{n} \cdots b_{1}}^{\:\:\:\:\:\:\:\: a'}(x), \\
\label{eqn:dottedtranslaws1} \dot{z}_{k}^{i'} &=& \dot{z}_{k}^{j} T_{j}^{\:\: i'}(x) + \sum_{w_{1}+ \cdots +  w_{n} =k} \frac{1}{(n-1)!}\:\: \dot{y}^{b_{1}}_{w_{1}} y^{b_{2}}_{w_{2}} \cdots y^{b_{n}}_{w_{n}} T_{b_{n} \cdots b_{1}}^{\:\:\:\:\:\:\:\: i'}(x),
\end{eqnarray}
where $1 \leq w< k$ and they are obtained by differentiation of the undotted coordinates.

It is important to note that the linearisation has the structure of a vector bundle $D(F_{k}) \rightarrow F_{k-1}$, hence the nomenclature ``linearisation".  The vector bundle structure is clear from the construction and transparent from the above admissible changes of local coordinates which are linear in dotted coordinates. Via inspection, we see that have (partially) polarized the polynomial coordinate changes (\ref{eqn:translaws}) to obtain a linear structure. That is via differentiation we have  increased the number of homogeneous coordinates in a consistent way in order to build a vector bundle structure.

\begin{example}
Considering the total space of a vector bundle $E$ as a graded bundle of degree one, we get $D(E) \simeq E$. The homogeneous coordinates on the vertical bundle are $(x,y, \dot{y})$. Then via the linearisation procedure we see that $D(E)$ comes with local coordinates $(x,\dot{y})$ and so can be identified with $E$.
\end{example}

\begin{example}
If $F_{2} = \sT^{2}M$, then $D(F_{2}) \simeq \sT(\sT M)$.
\end{example}

\subsection{Functorial properties}
Let us note that the construction of the linearisation has functorial properties.

\begin{theorem}\label{theom:functorial linearisation}
Linearisation is a functor from the category of graded bundles to the category of double graded bundles.
\end{theorem}
\begin{proof}
First we need to prove that given some morphism $\phi: F_{k} \rightarrow F_{\bar{k}}'$ between graded bundles then canonically there is a unique induced morphism of double graded bundles $D\phi : D(F_{k}) \rightarrow D(F_{\bar{k}}')$. To do this we will employ homogeneous local coordinates $(x^{A}, y^{a}_{w}, z^{i}_{k})$ on $F_{k}$ and $(\bar{x}^{\alpha}, \bar{y}^{\mu}_{\bar{w}}, \bar{z}^{p}_{\bar{k}})$  on $F_{\bar{k}}'$, where $1 \leq w < k$ and $1 \leq \bar{w} < \bar{k}$. Then the components of the morphism $\phi$ are
\begin{equation*}
\phi^{*}\bar{x}^{\alpha} = \phi^{\alpha}(x), \hspace{25pt} \phi^{*}\bar{y}_{\bar{w}}^{\mu} = \phi^{\mu}_{\bar{w}}(x,y,z), \hspace{25pt}\phi^{*}\bar{z}^{p}_{\bar{k}} = \phi^{p}_{\bar{k}}(x,y,z).
\end{equation*}
\noindent Note the the components of $\phi$ that depend on $y$ and $z$ are necessarily polynomial in these coordinates and are of the necessary weight. Next we lift $\phi$ to $V\phi : VF_{k} \rightarrow VF_{\bar{k}}'$. In obvious notation, this lifted map acts on the linear fibre coordinates as
\begin{eqnarray}
\nonumber (V\phi)^{*}\dot{\bar{y}}^{\mu}_{\bar{w}} &=& \dot{y}_{w}^{a} \frac{\partial \phi^{\mu}_{\bar{w}}}{\partial y^{a}_{w}}(x,y,z) + \dot{z}_{k}^{i} \frac{\partial \phi^{\mu}_{\bar{w}}}{\partial z^{i}_{k}}(x,y,z), \\
\nonumber (V\phi)^{*}\dot{\bar{z}}^{p}_{\bar{k}} &=& \dot{y}_{w}^{a} \frac{\partial \phi^{p}_{\bar{k}}}{\partial y^{a}_{w}}(x,y,z) + \dot{z}_{k}^{i} \frac{\partial \phi^{p}_{\bar{k}}}{\partial z^{i}_{k}}(x,y,z).
\end{eqnarray}
\noindent Just by inspection it is clear that we have a morphism in the category of double graded bundles. Passing to the respective linearisations is achieved via the canonical projections $VF_{k} \rightarrow D(F_{k})$ and $VF_{\bar{k}}' \rightarrow D(F_{\bar{k}}')$.

\begin{tabular}{p{5cm} p{10cm}}
\begin{center}
\leavevmode
\begin{xy}
(0,20)*+{D(F_{k})}="a"; (25,20)*+{D(F^{\prime}_{\bar{k}})}="b";%
(0,0)*+{F_{k-1}}="c"; (25,0)*+{F^{\prime}_{\bar{k}-1}}="d";%
{\ar "a";"b"}?*!/_3mm/{D\phi};%
{\ar "a";"c"}; {\ar "b";"d"};%
{\ar "c";"d"}?*!/^3mm/{\psi};%
\end{xy}
\end{center}
&
\vspace{15pt}
In the diagram the morphism $\psi:F_{k-1} \rightarrow F'_{\bar{k}-1}$ is the appropriate  restriction of $\phi$ induced by the obvious projections.
\end{tabular}

\noindent  Then in homogeneous local coordinates we have
 \begin{eqnarray}
 \nonumber \psi^{*}\bar{x}^{\alpha} &=& \phi^{\alpha}(x),\\
  \nonumber \psi^{*}\bar{y}_{\bar{w}}^{\mu} &=& \left. \phi^{\mu}_{\bar{w}}(x,y,z) \right|_{z},\\
  \nonumber (D\phi)^{*}\dot{\bar{y}}^{\mu}_{\bar{w}} &=& \dot{y}_{w}^{a} \left.\frac{\partial \phi^{\mu}_{\bar{w}}}{\partial y^{a}_{w}}(x,y,z)\right|_{z} + \dot{z}_{k}^{i} \left.\frac{\partial \phi^{\mu}_{\bar{w}}}{\partial z^{i}_{k}}(x,y,z)\right|_{z}, \\
\nonumber (D\phi)^{*}\dot{\bar{z}}^{p}_{\bar{k}} &=& \dot{y}_{w}^{a}\left. \frac{\partial \phi^{p}_{\bar{k}}}{\partial y^{a}_{w}}(x,y,z)\right|_{z} + \dot{z}_{k}^{i} \left.\frac{\partial \phi^{p}_{\bar{k}}}{\partial z^{i}_{k}}(x,y,z)\right|_{z},
\end{eqnarray}
where ``$|_{z}$" signifies that we project out the $z$ coordinate. Locally, we can view this as setting $z=0$.  Via inspection it is clear that we have a well defined morphism of double graded bundles. We now need to verify the expected properties of the linearisation.

\begin{itemize}
\item Clearly, we have that $D(\Id_{F}) = \Id_{D(F)}$.  Thus, $D$  respects the identities.
\item  The property $D(\phi \circ \chi) =  D(\phi) \circ D(\chi)$  where $\phi : F \rightarrow F'$ and $\chi : F' \rightarrow F''$ follows from the chain rule.
\end{itemize}
Thus $D$  is a functor.
\end{proof}

 Another important property of the linearisation describes the following theorem.
\begin{theorem}\label{theom:embedding}
There exists a canonical weight preserving  embedding of the graded bundle $F_{k}$
\begin{equation*}
\iota_{F_k} : F_{k} \hookrightarrow D(F_{k}),
\end{equation*}
 where we employ the total weight on $D(F_{k})$, given by the image of the weight vector field $\Delta_{F} \in \Vect(F)$ considered as a geometric section of $VF_{k}$. That is we have the following commutative diagram

\begin{center}
\leavevmode
\begin{xy}
(0,20)*+{F_{k}}="a"; (25,20)*+{V F_{k}}="b";%
(25,0)*+{D(F_{k})}="d";%
{\ar "a";"b"}?*!/_3mm/{\Delta_{F_{k}}};%
{\ar "a";"d"}?*!/^3mm/{\iota_{F_{k}}};%
{\ar "b";"d"}?*!/_7mm/{\textnormal{p}_{D(F_{k})}^{V F_{k}}};%
\end{xy}
\end{center}

\noindent In natural local coordinates the nontrivial part of the embedding  is given by
\begin{equation*}
\iota_{F_k}^{*}(\dot{y}^{a}_{w}) = w \: y^{a}_{w}, \hspace{25pt} \iota_{F_k}^{*}(\dot{z}_{k}^{i}) = k\: z_{k}^{i}.
\end{equation*}
\end{theorem}

\begin{proof}
 All maps are covariantly defined. It is clear via local coordinates that $\iota_{F_k}$ is a morphism of graded manifolds. Now, as $\Delta_{F_k}$ is a section of the vertical bundle, its image defines an embedding of $F_{k}$ in $VF_{k}$. The projection to $D(F_{k})$ then defines an embedding of $F_{k}$ in $VF_{k}$.
\end{proof}
\noindent
Elements of $\iota_{F_k}(F_{k})$ we will refer to as \emph{holonomic vectors} in $D(F_{k})$.

\begin{example}
It is worth considering a low degree example very explicitly. Consider $F_{3}$ equipped with natural coordinates $(x^{A}, y^{a}, z^{i}, w^{\kappa})$ of weight $0,1,2,3$ respectively. (We have a slight change in notation here, but that hopefully will not cause confusion.) The admissible transformation laws here are of the form;
\begin{align*}
& x^{A'}= x^{A'}(x), y^{a'} = y^{b}T_{b}^{a'}(x), & \\
& z^{i'} = z^{j}T_{j}^{\:\:i'}(x)+ \frac{1}{2!} y^{a}y^{b}T_{ba}^{i'}(x),& \\
& w^{\kappa'} = w^{\mu}T_{\mu}^{\:\: \kappa'}(x) + z^{i}y^{a}T_{ai}^{\:\:\: \kappa'}(x) + \frac{1}{3!}y^{a}y^{b}y^{c}T_{cba}^{\:\:\:\: \kappa'}(x).&
\end{align*}

\noindent On $D(F_{3})$ we have the local coordinates $(x^{A}, y^{a}, z^{i}, \dot{y}^{b}, \dot{z}^{j}, \dot{w}^{\kappa})$. The admissible transformation laws for the extra coordinates  are obtained by differentiation;
\begin{align*}
&\dot{y}^{a} = \dot{y}^{b}T_{b}^{a'}(x),&\\
&\dot{z}^{i'} = \dot{z}^{j}T_{j}^{\:\:i'}(x)+ \dot{y}^{a}y^{b}T_{ba}^{i'}(x),&\\
&\dot{w}^{\kappa'} = \dot{w}^{\mu}T_{\mu}^{\:\: \kappa'}(x) + \dot{z}^{i}y^{a}T_{ai}^{\:\:\: \kappa'}(x) + z^{i}\dot{y}^{a}T_{ai}^{\:\:\: \kappa'}(x)+  \frac{1}{2!}\dot{y}^{a}y^{b}y^{c}T_{cba}^{\:\:\:\: \kappa'}(x).&
\end{align*}

\noindent Now set $\dot{y}=y$, which clearly is consistent with the transformation laws. The transformation law for $\dot{z}$ becomes
\begin{equation*}
\dot{z}^{i'} = \dot{z}^{j}T_{j}^{\:\:i'}(x)+ y^{a}y^{b}T_{ba}^{i'}(x),
\end{equation*}
which is ``twice" that of  the transformation law for $z$. Thus we can set $\dot{z} = 2z$. Now using both the pervious identifications the transformation law for $\dot{w}$ becomes
\begin{equation*}
\dot{w}^{\kappa'} = \dot{w}^{\mu}T_{\mu}^{\:\: \kappa'}(x) + 3z^{i}y^{a}T_{ai}^{\:\:\: \kappa'}(x)  + \frac{1}{2!}y^{a}y^{b}y^{c}T_{cba}^{\:\:\:\: \kappa'}(x),
\end{equation*}
which is a ``thrice"  that of the transformation law for $w$. Thus we can set $\dot{w} = 3 w$. Then rather explicitly we see  that the numerical prefactor is just the weight of the respective coordinate and so this corresponds to the weight vector field on $F_{3}$. Thus we see that the embedding is indeed defined as the image of the weight vector field on $F_{3}$ considered as a geometric section.
\end{example}

\begin{example}
Let us consider $F_{2} = \sT E$, where $E$ is the total space of a vector bundle over $M$. Note here that we have the natural structure of a double vector bundle and so by taking the total weight we have a degree 2 graded bundle. Let us equip $\sT E$ with natural local coordinates $(x^{A}, y^{a}, \dot{x}^{B}, \dot{y}^{b})$. The admissible changes of local coordinates are of the form\\

\begin{tabular}{p{10cm} p{5cm}}
\begin{tabular}{l}
$x^{A'} = x^{A'}(x)$,  \hspace{15pt} $y^{a'} = y^{b}T_{b}^{\:\: a'}$,\\
$\dot{x}^{A'} = \dot{x}^{B}\frac{\partial x^{A'}}{\partial x^{B}}$, \hspace{15pt} $\dot{y}^{a'} = \dot{y}^{b}T_{b}^{\:\: a'} + y^{b}\dot{x}^{B} \frac{\partial T_{b}^{\:\: a'}}{\partial x^{B}}$.
\end{tabular}\\
\vspace{15pt}
\noindent which give the  double vector bundle structure
&
\vspace{-50pt}
\begin{center}
\leavevmode
\begin{xy}
(10,50)*+{\sT E}="a";%
(0,40)*+{E}="b"; (20,40)*+{\sT M}="c";%
(10,30)*+{ M}="d";%
{\ar "a";"b"};%
{\ar "a";"c"};%
{\ar "b";"d"};%
{\ar "c";"d"};%
\end{xy}
\end{center}
\end{tabular}

Now we construct $D( \sT E) \subset \sT(\sT E)$, which we equip with natural local coordinates $(x^{A}, y^{a}, \dot{x}^{B}, \delta y^{b}, \delta \dot{x}^{C}, \delta \dot{y}^{c})$. The admissible changes of local coordinates for the ``extra" coordinates are of the form

\vspace{10pt}
\begin{tabular}{p{10cm} p{5cm}}
\begin{tabular}{l}
$\delta y^{a'} = \delta y^{b}T_{b}^{\:\: a'}$,\hspace{15pt}$\delta \dot{x}^{A'} = \delta \dot{x}^{B}\frac{\partial x^{A'}}{\partial x^{B}}$, \\
$\delta \dot{y}^{a'} = \delta \dot{y}^{b}T_{b}^{\:\: a'} + \delta y^{b} \dot{x}^{B}\frac{\partial T_{b}^{\:\: a'}}{\partial x^{B}} +  y^{b} \delta \dot{x}^{B}\frac{\partial T_{b}^{\:\: a'}}{\partial x^{B}} $
\end{tabular}\\
\vspace{15pt}
\noindent which gives the  double vector bundle
&
\vspace{-50pt}
\begin{center}
\leavevmode
\begin{xy}
(10,50)*+{D(\sT E)}="a";%
(0,40)*+{E\times_{M} \sT M}="b"; (20,40)*+{E\times_{M} \sT M}="c";%
(10,30)*+{ M}="d";%
{\ar "a";"b"};%
{\ar "a";"c"};%
{\ar "b";"d"};%
{\ar "c";"d"};%
\end{xy}
\end{center}
\end{tabular}

We see that $D(\sT E) \subset \sT(\sT E)$ is a double vector bundle that closely resembles $\sT E$, but it is symmetric in its legs. From the coordinate transformations it is clear that the total weight vector field on $\sT E$, which is given by $\Delta_{\sT E} = y^{a}\frac{\partial }{\partial y^{a}} + \dot{x}^{A} \frac{\partial }{\partial x^{A}} + 2\: \dot{y}^{a} \frac{\partial}{\partial \dot{y}^{a}}$ produces the embedding $\sT E \hookrightarrow D(\sT E)$.
\end{example}

\subsection{The linearisation diagram}
A quick inspection shows that
\bea\label{linprop}
D(F_{k})[\zD^2\le 0]&\simeq&F_{k-1}\,,\\
D(F_{k})[\zD^1\le k-2]&\simeq&D(F_{k-1})\,,\nn\\
D(F_{k})[\zD\le k-1]&\simeq&VF_{k-1}\,,\nn
\eea
and that we get the following commutative diagram of graded bundle morphisms
\be\label{shortdiagram}
\xymatrix@R15mm @C15mm{V F_{k}\ar[d]_{d_k}\ar[rr]^{V\zt^{k-1}}&&V F_{k-1}\ar[d]^{d_{k-1}}\\
D(F_k)\ar[rr]^{D\zt^{k-1}} \ar[d]_{p_{k-1}}\ar[rru]^{v_{k-1}} && D(F_{k-1})\ar[d]^{{p_{k-2}}} \\
F_{k-1}\ar[rr]^{{\zt^{k-2}}}\ar[rru]^{\iota_{{k-1}}} && F_{k-2} }
\ee
Here,
$$d_k=\textnormal{p}^{VF_{k}}_{D(F_{k})}\,,\quad p_{k-1}=\textnormal{p}^{D(F_{k})}_{[\zD^2\le 0]}\,,\quad v_{k-1}=\textnormal{p}^{D(F_{k})}_{[\zD\le k-2]}\,,\quad\zi_{k-1}=\zi_{F_{k-1}}\,.
$$
Of course, with the use of the tower (\ref{eqn:fibrations}) of fibrations, this can be prolonged by induction, so we get the following.
\begin{theorem} Any graded bundle $F_k$ of degree $k$ gives rise to the commutative diagram of
graded bundle morphisms
\be\label{lindiagram}
\xymatrix@R15mm %@C15mm
{V F_{k}\ar[d]_{d_k}\ar[rr]^{V\zt^{k-1}}&&V F_{k-1}\ar[d]^{d_{k-1}}\ar[rr]^{V\zt^{k-2}}&&V F_{k-2}\ar[d]^{d_{k-2}}&\cdots&VF_2\ar[d]^{d_2}\ar[rr]^{V\zt^{1}}&&V F_{1}\ar[d]^{d_1}\\
D(F_k)\ar[rr]^{D\zt^{k-1}} \ar[d]_{p_{k-1}}\ar[rru]^{v_{k-1}} && D(F_{k-1})\ar[rr]^{D\zt^{k-2}} \ar[d]_{p_{k-2}}\ar[rru]^{v_{k-2}}&&
D(F_{k-2})\ar[d]^{p_{k-3}}&\cdots& D(F_2)\ar[d]^{p_1}\ar[rr]^{D\zt^{1}} \ar[rru]^{v_1}&&
D(F_{1})\ar[d]^{p_0} \\
F_{k-1}\ar[rr]^{{\zt^{k-2}}}\ar[rru]^{\iota_{{k-1}}} && F_{k-2}\ar[rru]^{\iota_{{k-2}}}\ar[rr]^{{\zt^{k-3}}} && F_{k-3}& \cdots& F_1\ar[rru]^{\iota_{1}}\ar[rr]^{{\zt^{0}}} && M }
\ee
\end{theorem}

 \subsection{$\mathcal{GL}$-bundles}
\noindent Mimicking the properties of the linearisation of a graded bundle consider an  abstract double graded bundle with one structure of degree $k-1$ and the other of degree $1$ (vector bundle).
\medskip
\begin{definition} A double graded bundle
$(D_{k}, \Delta^{1}, \Delta^{2})$ such that
$\zD^1$ is of degree $k-1$ and $\zD^2$ is of degree 1, i.e. $\Delta^{2}$ is an Euler vector field, will be referred to as a  \emph{graded--linear bundle} of degree $k$, or for short a $\mathcal{GL}$-bundle.
\end{definition}
\begin{remark}
Of course, one can define a $\mathcal{GL}$-bundle in terms of a pair of commuting homogeneity structures of degree $k-1$ and $1$. From this perspective, it is very clear that a  $\mathcal{GL}$-bundle is a vector bundle in the category of graded bundles or evidently vice versa.
\end{remark}
\noindent It follows that $D_{k}$ is of total weight $k$ with respect to $\Delta := \Delta^{1} + \Delta^{2}$ and a vector bundle structure
\begin{equation*}
\textnormal{p}^{D_{k}}_{B_{k-1}}: D_{k} \rightarrow B_{k-1}.
\end{equation*}
with respect to the projection onto the submanifold $B_{k-1} := D_{k}[\Delta^{2} \leq 0]$ which inherits a graded bundle structure of degree $k-1$.

\begin{tabular}{p{5cm} p{10cm}}
\begin{center}
\leavevmode
\begin{xy}
(0,20)*+{D_{k}}="a";  (30,20)*+{A_{k-1}}="b";%
(0,0)*+{B_{k-1}}="c"; %
{\ar "a";"b"}?*!/_4mm/{\textnormal{p}_{A_{k-1}}^{D_{k}}};%
{\ar "a";"c"}?*!/^7mm/{\textnormal{p}_{B_{k-1}}^{D_{k}}};%
{\ar "b";"c"}?*!/_3mm/{\textnormal{p}};%
\end{xy}
\end{center}
&
\vspace{20pt}
To set some further notation we define
$A_{k-1}:= D_{k}[\Delta \leq k-1]$. Note
that we have also
a vector bundle structure

$\textnormal{p}: A_{k-1} \rightarrow B_{k-1}$.
\end{tabular}

\begin{example}
The reader can easily verify that both $\sT F_{k-1}$ and $\sT^{*}F_{k-1}$ are $\mathcal{GL}$-bundles of total degree $k$, where $F_{k-1}$ is a graded bundle of degree $k-1$. Similarly, it is easy to see that a double vector bundle is a $\mathcal{GL}$-bundle of degree $2$.
\end{example}
\noindent
We require an internal characterisation of the linearisation of a graded bundle for later applications.
 It is clear from (\ref{linprop}) that a necessary condition for $D_k$ to be a linearisation of a graded bundle of degree $k$ is $A_{k-1}\simeq VB_{k-1}$.
More precisely, there exists a isomorphism $I: VB_{k-1} \rightarrow A_{k-1}$ such that the following diagram commutes:

\be\label{necessary}
\leavevmode
\begin{xy}
(0,20)*+{D_{k}}="a";  (30,20)*+{A_{k-1}}="b";%
(0,0)*+{B_{k-1}}="c"; (30,0)*+{VB_{k-1}}="d"; %
{\ar "a";"b"}?*!/_4mm/{\textnormal{p}_{A_{k-1}}^{D_{k}}};%
{\ar "a";"c"}?*!/^7mm/{\textnormal{p}_{B_{k-1}}^{D_{k}}};%
{\ar "b";"c"}?*!/_3mm/{\textnormal{p}};%
{\ar "d";"b"}?*!/^3mm/{I};%
{\ar "b";"c"}?*!/_3mm/{\textnormal{p}};%
{\ar "c";"d"}?*!/^3mm/{\Delta_{B_{k-1}}};%
\end{xy}
\ee
The above condition implies that we can always employ homogeneous local coordinates on $D_{k}$ of the form
\begin{equation*}
(\underbrace{x^{A}}_{(0,0)}, ~ \underbrace{y_{w}^{a}}_{(w,0)}, ~  \underbrace{\dot{y}_{w}^{b}}_{(w-1,1)}, ~ \underbrace{\bar{z}^{i}_{k}}_{(k-1,1)}).
\end{equation*}
The coordinate transformations for $x$ and $y$ are as standard for a degree $k-1$ graded bundle and the transformation law for $\dot{y}$ are of the form (\ref{eqn:dottedtranslaws}), that follows from the fact that we deal with the vertical bundle of $B_{k-1}$. We know that from the general structure of a double graded bundle that we must have transformation laws for $\bar{z}$ of the form
\begin{equation}\label{cc}
\bar{z}_{k}^{i'} = \bar{z}_{k}^{j}T_{j}^{\:\:i'}(x) + \sum_{\begin{tiny} w_{1} + \cdots w_{n} =k\end{tiny}} \frac{1}{(n-1)!}\dot{y}_{w_{1}}^{b_{1}}y_{w_{2}}^{b_{2}} \cdots y_{w_{l}}^{b_{n}}T_{b_{1}b_{2} \cdots b_{n}}^{\:\:\:\:\: \:\:\:\:\:\:\:\: i'}(x),
\end{equation}
where the combinatorial prefactor in the second term is chosen for later convenience.
They are of the form (\ref{eqn:dottedtranslaws1}), i.e. are obtained by differentiation of the undotted coordinates if and only if the tensors in the second term are symmetric in lower indices.
A $\mathcal{GL}$-bundle $D_k$ equipped with the isomorphism $I$ as in (\ref{necessary}) and an atlas such that the tensors $T_{b_{1}b_{2} \cdots b_{n}}^{\:\:\:\:\: \:\:\:\:\:\:\:\: i'}$
appearing in the coordinate change (\ref{cc}) are symmetric in lower indices will be called
\emph{symmetric}. For a symmetric $\mathcal{GL}$-bundle $D_k$ of degree $k$, the undotted coordinate changes define a graded bundle $F_k$ of degree $k$.
It is easy to see that $F_{k}$ is the submanifold in $D_{k}$ locally defined by $\dot{y}^{a}_{w} = w y_{w}^{a}$.  Thus, we can use coordinates $(x^{A}, y_{w}^{a}, \dot{y}_{w}^{b} , z_{k}^{i})$ on $F_{k}$, where $z^{i}_{k} = \frac{1}{k} \bar{z}^{i}_{k}$. It follows that $D_{k} \simeq D(F_{k})$ and that $\iota_{F_k}: F_{k} \hookrightarrow D(F_{k})$ corresponds to the inclusion $F_{k} \hookrightarrow D_{k}$. Thus, we get the following.

\begin{theorem}\label{theom:linearisation}
A $\mathcal{GL}$-bundle $D_{k}$ is the linearisation of a graded bundle $F_{k}$ if and only if
it is symmetric.
In this case, the graded bundle $F_{k}$ can be identified with the graded subbundle of ``holonomic vectors"  viz
\begin{equation*}
F_{k} \simeq (\textnormal{p}^{D_{k}}_{A_{k-1}})^{-1}\left( I ( \Delta_{B_{k-1}}(B_{k-1})) \right),
\end{equation*}
where $\Delta_{B_{k-1}} \in \Vect(B_{k-1})$ is the total weight vector field on $B_{k-1}$. The canonical embedding $\iota_{F_k} : F_{k} \hookrightarrow D_{k}$ is identified with the natural inclusion.
 \end{theorem}
\begin{remark}
One can introduce a coordinate free definition of \emph{category of symmetric $\mathcal{GL}$-bundles of degree $k$} and to show its equivalence with the category of graded manifolds of degree $k$. Doing the linearisation inductively, one can also replace the category of symmetric $\mathcal{GL}$-bundles of degree $k$ with a certain category of \emph{symmetric $k$-tuple vector bundles}. For $k=2$, this yields a result similar to that of Jotz Lean \cite{JotzLean:2015}. A detailed study of such question is the subject of a separate paper \cite{Bruce:2015}.
\end{remark}

\begin{example}
The above theorem mimics the well known case of $\sT^{2}M$ being the affine bundle of holonomic vectors in $\sT (\sT M)$.
\end{example}

\begin{example}
As a counterexample, consider $\sT E$ where $E$  is a vector bundle. The reader can quickly convince themselves  that  $E\times_{M} \sT M \neq VE$ unless $E = \sT M$ and thus the tangent bundle of a general vector bundle is \emph{not} the linearisation of a  graded bundle of degree 2.
 \end{example}
\begin{example}\label{reduction}
Similarly to graded linear bundles we can consider \emph{graded principal bundles}, being double structures consisting of a $G$-principal bundle $P_k\to N$ equipped additionally with a compatible structure of a graded bundle of degree $k$. The latter is encoded by a weight vector field $\nabla$ (or the corresponding
action $t\mapsto h_t$ of multiplicative reals) and  compatibility clearly means that the $G$-action and the $\R$-action commute, i.e. the weight vector field $\nabla$ is $G$-invariant. This implies, in particular, that $N$ carries itself a (reduced) graded bundle structure, and that the action of $G$ on $P_k$ projects to its actions on $P_{k-1},\dots,P_0=M$.

If the $G$-action on $P_0$ is still free, we say that the graded principal bundle is \emph{basic}.
An example of a basic graded principal bundle with structure group $G$ is $\sT^kG$ on which $g\in G$ acts by $\sT^k g$ (see section \ref{group-red}).

\end{example}
A direct consequence of theorem \ref{theom:functorial linearisation} is that any graded $G$-principal structure on $P_k$ induces a (bi-)graded $G$-principal structure on its linearization, i.e. $D(P_k)\to D(P_k)/G$ is a principal bundle. One can easily prove the following.
\begin{theorem}\label{linearisation-reduction}
If $P_k$ is a basic graded principal bundle with the structure group $G$, then
$$D(P_k/G)\simeq D(P_k)/G\,,$$
i.e. the linearisation of the reduced bundle $N=P_k/G$ is canonically isomorphic with the reduction  $D(P_k)/G$ of the linearisation $D(P_k)$.
\end{theorem}

\subsection{The linear dual of a graded bundle}
The double graded  bundle $D(F_{k})$ is  a vector bundle over $F_{k-1}$ and so one can thus employ standard linear notions such as the dual.

\begin{definition}
The \emph{linear dual of a graded bundle} $F_{k}$ is the dual of the vector bundle $D(F_{k}) \rightarrow F_{k-1}$, and we will denote this $D^{*}(F_{k})$.
\end{definition}

\begin{proposition}
\begin{equation*}
D^{*}(F_{k}) \simeq \sT^{*}F_{k}[\Delta^{1}_{\sT^{*}F_{k}} \leq k-1],
\end{equation*}
\noindent which gives the canonical projection $\sT^{*}F_{k} \rightarrow D^{*}(F_{k})$.
\end{proposition}

\begin{proof}
The proof is a standard exercise  in finding the form of the dual of the projection $\textnormal{p}^{VF_{k}}_{D(F_{k})} : VF_{k} \rightarrow D(F_{k})$ for local coordinates. Details are left to the reader.
\end{proof}
\noindent
In simpler terms, we can employ homogeneous coordinates inherited from  the cotangent bundle as
  \begin{equation*}
 (\underbrace{x^{A}}_{(0,0)} , ~ \underbrace{y_{w}^{a}}_{(w,0)}; ~ \underbrace{ \pi_{b}^{k-w+1}}_{{(k-w, 1)}} , \underbrace{\pi_{i}^{1}}_{(0,1)}),
 \end{equation*}

\noindent which have admissible changes of local (fibre) coordinates of the form

\begin{eqnarray}
\nonumber \pi_{a'}^{k-w+1} &=&  T_{a'}^{\:\: b}\pi_{b}^{k-w+1}  \:\:    {-}    \sum_{   \tiny\begin{array}{l}  w_{1} + \cdots\\ + w_{l}  = w'+w\end{array}} \frac{1}{l!}\:\: y_{w_{1}}^{b_{1}}\cdots y_{w_{l}}^{b_{l}}T_{a'}^{\:\: d}T_{d \: b_{l} \cdots b_{1} }^{\:\:\: \:\:\:\:\:\:\:\:\:c'}T_{c'}^{\:\: e}\pi_{e}^{k-w'+1}\\
\nonumber &{}& -\sum_{   \tiny\begin{array}{l}  w_{1} + \cdots\\ + w_{l} = k\end{array}} \frac{1}{l!}\:\: y_{w_{1}}^{b_{1}}\cdots y_{w_{l}}^{b_{l}}T_{a'}^{\:\: d}T_{d \: b_{l} \cdots b_{1} }^{\:\:\: \:\:\:\:\:\:\:\:\:i'}T_{i'}^{\:\: j}\pi^{1}_{j},\\
 \pi_{i'}^{1} &=& T_{i'}^{\:\: j}\pi_{j}^{1}.
\end{eqnarray}

Note that by passing to total weight the linear dual of a graded bundle of degree $k$ is itself a graded bundle of degree $k$. By inspection of the local coordinates and their transformation rules it is clear that $D^{*}(F_{k})$ is a $\mathcal{GL}$-bundle.

\begin{example}
If $F_{2} := \sT^{2}M$, then  $D(\sT^{2}M) \simeq \sT(\sT M)$. The linear dual here is $D^{*}(\sT^{2}M) \simeq \sT^{*}(\sT M)$  as expected from the theory of double vector bundles.
\end{example}

As standard we have the duality between $D(F_{k})$ and $D^{*}(F_{k})$ as vector bundles over $F_{k-1}$, which can be understood geometrically via $\sT^{*}D(F_{k}) \simeq \sT^{*}D^{*}(F_{k})$. More importantly, the natural pairing of the (linear)  fibre coordinates of $D(F_{k})$ and $D^{*}(F_{k})$ given by $\sum \pi^{k-w+1}_{a}\dot{y}^{a}_{w} + \pi_{i}^{1}\dot{z}^{i}_{k}$  induces a pairing of the fibre coordinates of $F_{k}\rightarrow M $ and $D^{*}(F_{k}) \rightarrow F_{k-1}$ via the canonical embedding  $\iota_{F_k} : F_{k} \longhookrightarrow D(F_{k})$.

\begin{proposition} There is a natural `pairing'
\be
\delta_{F_{k}}: D^{*}(F_{k}) \times_{F_{k-1}} F_{k} \longrightarrow  \mathbb{R}[k]
\ee
given in local coordinates by
$$\delta_{F_{k}}^{*}(t)  = \sum w \: \pi^{k-w+1}_{a} y^{a}_{w} + k \: \pi_{i}^{1}z_{k}^{i}\,,
$$
\noindent where we assign weight $k$ to the global coordinate $t$ on $\mathbb{R}[k]$.
\end{proposition}

\begin{remark} As well as the above natural pairing $\delta_{F_{k}}$, the linear dual of a graded bundle is a \emph{dual object} in the sense that one can pass from the graded bundle to its linear dual and vice versa ``without loss of information". Starting from the dual $D^{*}(F_{k})$ one canonically passes to $D(F_{k})$ as standard, and then via Theorem \ref{theom:linearisation} one recovers $F_{k}$ using the holonomic vectors.  However, the linear dual is  not completely satisfactory as an object dual to a graded bundle, since this `dualisation' is clearly not reflexive. This was the reason, why we called it `linear dual' and not just `dual'. None the less, the linear dual plays an important role in applications cf. \cite{Bruce:2014b}.
\end{remark}
There is a substructure of the linear dual of a graded bundle  that, up to our knowledge, first  appeared in the work of  Miron \cite{Miron:2003} for the case of higher order tangent bundles in the context of higher order Hamiltonian mechanics.

\begin{definition}
The \emph{Mironian} of a graded bundle is the double graded bundle defined as
\begin{equation*}
\textnormal{Mi}(F_{k}) := D^{*}(F_{k})[\Delta_{D^{*}(F)}^{1} + k\: \Delta_{D^{*}(F)}^{2} \leq k]
\end{equation*}
\end{definition}

\noindent It is not hard to see that the local coordinates on the Mironian inherited from the coordinates on the linear dual are

\begin{equation*}
 (\underbrace{x^{A}}_{(0,0)} , ~ \underbrace{y_{w}^{a}}_{(w,0)}; ~ \underbrace{\pi_{i}^{1}}_{(0,1)}),
 \end{equation*}

 \noindent and so the Mironian of $F_{k}$ has the structure of a vector bundle over $F_{k-1}$. Moreover, we have the identification $\textnormal{Mi}(F_{k}) \simeq F_{k-1} \times_{M} \bar{F}_{k}^{*}$.

\begin{example}
If $F_{2} =  \sT^{2}M$, then $\textnormal{Mi}(\sT^{2}M) =  \sT M \times_{M}\sT^{*}M $.
\end{example}

A natural question here is \emph{can one recover $F_{k}$ from $\textnormal{Mi}(F_{k})$?} As the Mironian of $F_{k}$ is a vector bundle over $F_{k-1}$, we can take the standard dual and obtain

\begin{equation*}
\textnormal{Mi}^{*}(F_{k}) = F_{k-1} \times_{M}\bar{F}_{k} \simeq F_{k},
\end{equation*}

\noindent where we have made use of the Gaw\c{e}dzki--Bachelor-like theorem for graded bundles, see Proposition \ref{prop:partialsplitting}. It seems we can recover $F_{k}$ from just $\textnormal{Mi}(F_{k})$, however the identification of $\textnormal{Mi}^{*}(F_{k})$  with $F_{k}$ is non-canonical and thus the identification is really only up to isomorphism classes.  In simpler language, the Mironian misses the full transformation law for the highest weight coordinates on $F_{k}$. Thus we can never fully recover the initial graded bundle from the Mironian without the extra data of a partial splitting. Therefore, the Mironian, like the linear dual,  is \emph{not} a completely satisfactory notion of a dual object to a graded bundle. The Mironian also plays an important role in geometric mechanics.

Amongst all the possible fibrations of $\sT^{*}F_{k}$, there is a privileged role for

\vspace{10pt}
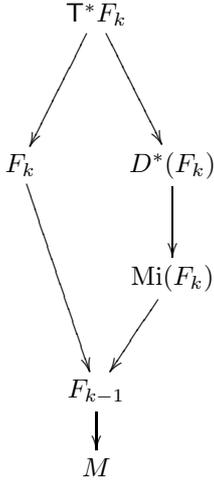
 \begin{tabular}{p{5cm} p{10cm}}
\leavevmode
\begin{center}
\begin{xy}
(10,60)*+{\sT^{*}F_{k}}="a"; %
(0,40)*+{F_{k}}="b";  (20,40)*+{D^{*}(F_{k})}="c";%
(20,25)*+{\textnormal{Mi}(F_{k})}="d";%
(10,10)*+{F_{k-1}}="e";%
(10,0)*+{M}="f";%
{\ar "a";"b"}; {\ar "a";"c"}; %
{\ar "b";"e"}; {\ar "d";"e"};%
{\ar "c";"d"};%
{\ar "e";"f"};%
\end{xy}
\end{center}
&
\vspace{10pt}
in the sense that the vector bundle on the right-hand side of the diagram is the right object to play the r\^ole of a dual of a graded bundle. Note that the dimension of $F_{k}$ and $D^{*}(F_{k})$ are different in general, but the dimension of the typical fibres over $M$ and $F_{k-1}$ respectively are equal.
\end{tabular}

\subsection{Application of the parity reversion functor}
Superising graded  and $n$-tuple graded bundles beyond the linear case of vector and $n$-tuple vector bundles is a non-trivial task with as of today no systematic procedure. In fact, it is far from obvious that an arbitrary  graded bundle has a well-defined superisation.  The linearisation of a graded bundle allows us to define a ``partial superisation" of the graded bundle (see \cite{Bruce:2015} where this procedure is iterated). In particular, as both $D(F_{k})$ and $D^{*}(F_{k})$ are vector bundles over $F_{k-1} $ with some additional weights, one can directly employ the parity reversion functor to consistently define $\Pi D(F_{k})$ and $\Pi D^{*}(F_{k})$, which are then graded super bundles in our language. That is the coordinates now carry both weight and Grassmann parity.  The parity reversion of the Mironian can also be directly constructed.

Alternatively, one could define  $\Pi D(F_{k})$ and $\Pi D^{*}(F_{k})$ as reductions of $\Pi \sT F_{k}$ and $\Pi \sT^{*}F_{k}$ following the constructions presented earlier in this section.

\begin{example}
If $F_{2} = \sT^{2}M$, then $\Pi D(\sT^{2}M) \simeq \Pi \sT (\sT M)$ and $\Pi D^{*}(\sT^{2}M) \simeq \Pi \sT^{*} (\sT M)$. Furthermore, we have $\Pi \textnormal{Mi}(\sT^{2}M) \simeq \sT M \times_{M} \Pi \sT^{*} M$.
\end{example}

\begin{example}\label{exm:odd dvb}
To illustrate the more general situation consider an arbitrary graded bundle of degree 2. The supermanifold  $\Pi D(F_{2})$ comes with natural homogeneous coordinates
\begin{equation*}
\{\underbrace{x^{A}}_{(0,0)}, ~ \underbrace{y^{a}}_{(1,0)}; \vspace{5pt} \underbrace{\zx^{b}}_{(0,1)}, ~ \underbrace{\theta^{i}}_{(1,1)}  \},
\end{equation*}
 where the last two coordinates are Grassmann odd. The admissible changes of local coordinates are of the form
\begin{align*}
&x^{A'} = x^{A'}(x),  & & y^{a'} = y^{b}T_{b}^{\:\: a'}(x),&\\
&\zx^{a'} = \zx^{b}T_{b}^{\:\: a'}(x), &  & \theta^{i'} = \theta^{j}T_{j}^{\:\: i'}(x) + \zx^{a}y^{b}T_{ba}^{\:\:\: i'}(x).&
\end{align*}
\end{example}

\begin{remark}
The multi-graded supermanifolds constructed above, and employed in the rest of this paper, are not quite as general as the non-negatively graded (super)manifolds  defined by Voronov \cite{Voronov:2001qf}. In particular we see that the supermanifolds here arise via the application of the parity reversion functor which acts on linear coordinates only. Thus, the Grassmann parity and the weights are not completely independent: the Grassmann parity will be specified by one component of multi-weights employed which will be either zero or one.
\end{remark}

Also any graded--linear bundle admits a parity reversion over the linear structure. That is $\Pi D_{k} \rightarrow B_{k-1}$ makes sense and this induces $\Pi A_{k-1} \rightarrow B_{k-1}$  in a consistent manner. However, there is no obvious total parity reversion unless we are specifically dealing with a double vector bundle structure.

The linearisation and linear dual of a graded bundle do allow us to employ ``linear constructions" in the theory of graded bundles. In particular they will be fundamental in defining weighted algebroids in the next section.

\section{Weighted algebroids}\label{sec:weighted algebroids}
\subsection{Weighted algebroids; the general structure}
The approach developed here is to generalise the double vector bundle morphism $\epsilon: \sT^{\ast} E \rightarrow \sT E^{\ast}$ covering the identity on $E^{\ast}$ that encodes the Lie algebroid structure on  vector  bundle $E \rightarrow M$. A more general double vector bundle morphism of this type is known as a \emph{general algebroid} \cite{Grabowski:1999}. The main motivation for defining Lie algebroid-like objects in terms of double vector morphisms is the fact that such morphisms are the key in developing geometric mechanics.

Recall that  $D(F_{k}) \rightarrow F_{k-1}$ is  a vector bundle structure carrying some additional weights. We wish to use this linear structure to mimic the standard constructions in the theory of general algebroids to define a higher analogue of a Lie algebroid on $F_{k}$. In order to obtain slightly more general algebroid structures we will start from an arbitrary  $\mathcal{GL}$-bundle and then later pass to the special case of the linearisation of a graded bundle.

\begin{definition}\label{def:GLalgebroid}
 A \emph{weighted algebroid of degree k} is a morphisms of triple graded bundles
 \begin{equation}\label{epsilon}
 \epsilon : \sT^{\ast} D_{k} \rightarrow \sT D^{\ast}_{k}
 \end{equation}
  \noindent covering the identity on the double graded bundle $D^{\ast}_{k}$.  The \emph{anchor map} of a weighted algebroid is the map $\rho : D_{k} \rightarrow \sT B_{k-1}$ underlying the map $\epsilon$, see the diagram below.
\end{definition}

\begin{equation*}\xymatrix{
  \sT^\ast D_k \ar[rr]^{\epsilon} \ar[dr]
 \ar[dd]
 & & \sT D^*_k\ar[dr]\ar[dd]
 & \\
  & D_{k}\ar[rr]^/-20pt/{\zr}\ar[dd]
 & &  \sT B_{k-1}\ar[dd]\\
 D^\ast_{k}\ar[rr]^/+20pt/{id}\ar[dr]
 & & D^\ast_{k}\ar[dr] &   \\
 &B_{k-1}\ar[rr]^{id} & & B_{k-1}
}
\end{equation*}

\noindent Recall that we have defined $B_{k-1} :=  D_{k}[\Delta^{2} \leq 0]$, which by assumption is a graded bundle of degree $k-1$. We will denote a weighted algebroid  as a pair $(D_{k}, \epsilon)$, where the graded structure described by the weight vector fields is understood.

\begin{remark}
In a more categorical language, our definition of a weighted algebroid is equivalent to a general algebroid in the category of graded bundles. In particular, we see that the structure defining map $\epsilon$ by definition is a morphism of graded bundles by forgetting the linear structures.
\end{remark}

The morphism (\ref{epsilon}) is known to be associated with a 2-contravariant tensor field $\Lambda_{\epsilon}$ on $D_k^*$. If  $\Lambda_{\epsilon}$ associated with $\epsilon$  is a bi-vector field on $D^{*}_{k}$, we speak about a \emph{weighted skew algebroid of degree k}. If the bi-vector field is a Poisson structure then we have a \emph{weighted Lie algebroid of degree k}.

\begin{remark}
As $B_{k-1}$ is a graded bundle, we can use the differential of the natural projections to $B_{k-l}$ for $1< l \leq k$ to define a series of anchors  $\zr_{k-l+1}:D_k\to \sT B_{k-l}$. In particular, we can define an anchor map $\zr_1:D_{k} \rightarrow \sT M$.  We will explore this properly for  weighted algebroids whose underlying structure is a linearisation of a graded bundle.
\end{remark}

We will primarily focus on weighted skew and Lie algebroids from now on.

\subsection{Weighted skew algebroids and almost Poisson structures}

In order to make the constructions explicit, we will  employ more compact notation with local coordinates. In particular we will choose homogeneous local coordinates
$$(\underbrace{x^{\alpha}_{u}}_{(u,0,0)},~\underbrace{y^{I}_{u+1}}_{(u,1,0)},~ \underbrace{p^{u+2}_{\beta}}_{(u,1,1)},~ \underbrace{\pi^{u+1}_{J}}_{(u,0,1)} ),$$
on $\sT^{*}D_{k}$. Here the $0 \leq u < k$. Similarly, we choose homogeneous local coordinates
$$(\underbrace{x^{\alpha}_{u}}_{(u,0,0)},~ \underbrace{\pi^{u+1}_{I}}_{(u,0,1)},~ \underbrace{\delta x^{\beta}_{u+1}}_{(u,1,0)},~ \underbrace{\delta \pi^{u+2}_{J}}_{(u,1,1)} ),$$
on $\sT D_{k}^{*}$.

\begin{proposition}\label{prop:bivectors}
There is a one-to-one correspondence between weighted  skew algebroids structures on $D_{k}$ and
\begin{enumerate}
\item
 bi-vector fields $\Lambda_{\epsilon}$ on $D^{*}_{k}$ of bi-weight $(1-k,-1)$;
\item skew algebroid brackets $[\cdot,\cdot]_\ze$ on the vector bundle  $D_k\to B_{k-1}$ which are
of degree $-k$, i.e. the bracket of homogeneous sections of degrees $r_1,r_2$ is of degree $r_1+r_2-k$;
\item homogeneous `Hamiltonians' $\sP_{\ze}$
of weight $(k-1,2, 1)$  on the graded super bundle $\Pi \sT^{*}D^{*}_{k}$.
 \end{enumerate}
\end{proposition}
 \begin{proof}
 It is well known that there is a one-to-one correspondence between skew algebroid structures on vector bundles and linear bi-vector fields on the dual vector bundle \cite{Grabowski:1999}. The only real complication in the current situation is the presence of the extra weights in addition to the weight associated with the linear structures. Thus we will not repeat the classical proofs in full, rather we will concentrate on ensuring the various weights are correct.\medskip

 \noindent It is well known that given a linear bi-vector field $\Lambda_{\epsilon}$ gives rise to the vector bundle morphisms

 \begin{equation*}
 \bar{\Lambda}_{\epsilon}: \sT^{*}D^{*}_{k} \longrightarrow \sT D^{*}_{k},
\end{equation*}

\noindent which is induced by the the left interior product of a one-form, that is a section of $\sT^{*}D^{*}_{k}$ and the bi-vector field. The fact that the bi-vector field is linear translates into the graded language as the statement that the bi-vector field is of bi-weight $(\ast, -1)$, where $\ast$ is to shortly be determined.  As the morphism  $\bar{\Lambda}_{\epsilon}$ also preserves total weight, we know  that
\begin{equation*}
\bar{\Lambda}_{\epsilon}^{*}(\delta x^{\alpha}_{u+1}) = y^{I}_{u-u'+1}P_{I}^{\alpha}[u'](x),
\end{equation*}
\noindent where $P_{I}^{\alpha}[u'](x)$ is of tri-weight $(u', 0,0)$ and necessarily polynomial in coordinates of non-zero tri-weight. In the above expression, and similar expressions, any coordinates such that the label falls outside of the range of the total weight are set to zero.

 This means that the bi-vector field must be of the form
\begin{equation*}
\Lambda_{\epsilon} = P^{\alpha}_{I}[u'](x)  \frac{\partial }{\partial \pi_{I}^{\rlap{$\scriptstyle k+u'-u$}}}\hspace{25pt}  \wedge \frac{\partial}{\partial x^{\alpha}_{u}}  + \textnormal{More},
\end{equation*}
\noindent remembering the shift in the weight required when defining the cotangent bundle of a graded bundle. Thus, we see that $\Lambda_{\epsilon}$ is of total weight $-k$ and thus of bi-weight $(1-k,-1)$. This establishes 1.\medskip

\noindent In homogeneous local coordinates  the bi-vector field encoding a weighted skew algebroid is of the form
\begin{equation}
\Lambda_{\epsilon} = P^{\alpha}_{I}[u'](x)  \frac{\partial }{\partial \pi_{I}^{\rlap{$\scriptstyle k+u'-u$}}}\hspace{25pt}  \wedge \frac{\partial}{\partial x^{\alpha}_{u}}   + \frac{1}{2!} P_{IJ}^{K}[u'](x)\pi_{K}^{k-u} \frac{\partial }{\partial \pi_{J}^{\rlap{$\scriptstyle k+u''-u$}}}\hspace{25pt}  \wedge \frac{\partial }{\partial \pi_{I}^{\rlap{$\scriptstyle k+u'-u''$}}} \hspace{30pt}.
\end{equation}

\medskip

\noindent It will be convenient at this point to first prove 3. before 2. It is well known that there is a one-to-one correspondence between multivector fields on a manifold $N$ and functions on the supermanifold $\Pi T^{*}N$ and that the classical Schouten--Nijenhuis bracket is in correspondence with the Schouten (or odd Poisson) bracket associated with the odd symplectic structure on  $\Pi T^{*}N$.  Let us on $\Pi \sT^{*}D_{k}^{*}$  employ local homogeneous coordinates

\begin{equation*}
 (\underbrace{x^{\alpha}_{u}}_{(u,0,0)}, ~  \underbrace{\pi^{u+1}_{I}}_{(u,0,1)}, ~ \underbrace{\chi_{\beta}^{u+2}}_{(u,1,1)}, ~ \underbrace{\theta_{u+1}^{J}}_{(u,1,0)}).
\end{equation*}

\noindent Note that the assignment of the Grassmann parity is encoded in the assignment of the second component of the tri-weight in our convention. The association of a `Hamiltonian' with a multi-vector field is viz
\begin{align*}
& \chi_{\alpha}^{k+1-u} \longleftrightarrow \frac{\partial}{\partial x^{\alpha}_{u}}, & &   \theta_{u-u'+1}^{I} \longleftrightarrow \frac{\partial }{\partial \pi_{I}^{\rlap{$\scriptstyle k+u'-u$}}}\hspace{25pt},&
\end{align*}

\noindent noting the shift in weight \emph{and} Grassmann parity. Essentially one exchanges the wedge product of partial derivatives for the supercommutative product of the antimomenta to obtain

\begin{equation}
 \sP_{\ze} =  \theta^{I}_{u-u'+1}P^{\alpha}_{I}[u'](x) \chi_{\alpha}^{k+1-u} + \frac{1}{2!} \theta^{J}_{u-u''+1} \theta^{I}_{u''-u'+1} P_{IJ}^{K}[u'](x)\pi_{K}^{k-u},
\end{equation}

\noindent understood in the super-language as a function on $\Pi \sT^{*}(D^{*}_{k})$. Via inspection we see that $\sP_{\ze}$ is of tri-weight $(k-1,2,1)$. This establishes 3.\medskip

\noindent To establish 2.,  note that we can  interpret sections of degree $r$ of   $D_{k} \rightarrow B_{k-1}$ as functions on $\Pi T^{*}D_{k}^{*}$ of tri-weight $(r-1,0,1)$. Furthermore note that the canonical Schouten bracket on  $\Pi T^{*}D_{k}^{*}$ is of tri-weight $(1-k,-1,-1)$. Then using the derived bracket formalism:

\begin{equation*}
[s_{1}, s_{2}]_{\epsilon} =  \SN{  \SN{ s_{1},\sP_{\ze}}, s_{2}},
\end{equation*}

\noindent which is understood at this stage as a function on $\Pi T^{*}D_{k}^{*}$, is of tri-weight $(r_{1}+ r_{2}-k-1,0,1)$, where $s_{1}$ is a section of degree $r_{1}$ and $s_{2}$ is a section of degree $r_{2}$. There is a little leeway here with conventions and thus we will intentionally be slack with an overall sign. Note that the skew bracket on sections of   $D_{k} \rightarrow B_{k-1}$ closes and thus we can interpret $[s_{1}, s_{2}]_{\epsilon}$  as a skew algebroid bracket of weight $-k$. As standard, the existence of the anchor follows from the Leibniz rule of the almost Poisson bracket. Thus, we have established 2. If $\sP_{\ze}$ is Poisson, that is $\SN{\sP_{\ze},\sP_{\ze}}=0$ then the associated skew algebroid bracket is  of course a Lie bracket.
\end{proof}

\begin{example}\label{Atiyah}
There is a well known construction of the Lie algebroid $\A(P)$ associated to a $G$-principal bundle $P\to N$, called usually the \emph{Atiyah algebroid}. The sections of the corresponding vector bundle $\A(P)\to N$, usually denoted by $\sT P/G$, consist of $G$-invariant vector fields on $P$, and the Lie algebroid bracket is just the Lie bracket of vector fields. Dually, $\A^*(P)$ is the Poisson reduction of the cotangent bundle $\sT^*P$ with respect to the lifted action of $G$. If $P_k$ is a graded principal bundle, then $\A(P_k)=\sT P_k/G$ inherits a graded bundle structure from $\sT P_k$, which is compatible with the Lie algebroid structure, i.e. we get an example of a weighted Lie algebroid. Indeed, according to our weight-convention for coordinates on $\sT^*P_k$, degree $r$ sections of the vector bundle $\A(P_k)\to N$ correspond to $G$-invariant vector fields $X$ on $P_k$ which are of degree $r-k$ with respect to the weight vector field $\nabla$ on $P_k$, $[\nabla,X]=(r-k)X$, where $[\cdot,\cdot]$ is the Lie bracket of vector fields. Since $[\nabla,X_i]=(r_i-k)X_i$, $i=1,2$, we get  from the Jacobi identity
$$[\nabla,[X_1,X_2]]=(r_1+r_2-2k)[X_1,X_2]=\left((r_1+r_2-k)-k\right)[X_1,X_2]\,,$$
i.e. the bracket of sections of degrees $r_1,r_2$ is of degree $r_1+r_2-k$, thus is itself of degree $-k$.
\end{example}

The  morphism $\epsilon: \sT^{*}D_{k} \rightarrow \sT D^{*}_{k}$ that defines the weighted skew algebroid  can be deduced from a given bi-vector field by examining  $\bar{\Lambda}_{\epsilon}$ and the canonical symplectomorphism $\sT^{*}D_{k} \rightarrow \sT^{*}D_{k}^{*}$.   One can more-or-less read off the desired morphism following \cite{Grabowski:1999}
\begin{eqnarray}
\delta x_{u+1}^{\alpha} \circ \epsilon  &=& y^{I}_{u-u'+1} P^{\alpha}_{I}[u'](x)\\
\nonumber \delta \pi_{J}^{u+1}\circ \epsilon &=& P^{\alpha}_{J}[u'](x)p_{\alpha}^{u-u'+2} + y^{I}_{u''-u'+1}P_{IJ}^{K}[u'](x)\pi_{K}^{u-u''+1}.
\end{eqnarray}

\subsection{Weighted Lie algebroids as homological vector fields}
As $(D_{k}, \Delta^{1}, \Delta^{2})$ has the structure of a vector bundle, there exists a canonical (odd) symplectomorphism $R: \Pi \sT^{*}\Pi D_{k} \rightarrow \Pi \sT^{*}D_{k}^{*}$, see for example \cite{Bruce:2011} for details. We can then use this canonical symplectomorphism to pull-back  the `Hamiltonian' $\sP_{\epsilon}$ on $\Pi T^{*}D^{*}_{k}$ encoding a weighted skew algebroid to a function linear in odd momenta on $\Pi \sT^{*}\Pi D_{k}$. This linear function  can be interpreted as the (odd) symbol of an odd   vector field on $\Pi D_{k}$. As we have a symplectomorphism and the symbol maps the Lie bracket of vector fields to the Schouten bracket of the symbols, the odd vector field associated with a weighted Lie algebroid is in fact a homological vector field. The remarkable point is that this homological  vector field on $\Pi D_{k}$ is naturally of total weight one, thus weighted Lie algebroids are very similar to non-linear or higher Lie algebroid as defined by Voronov \cite{voronov-2010}. We again stress that the weight one condition is not imposed on the odd vector field, rather it comes naturally from the definition of a weighted skew algebroid in terms of graded morphisms.

\begin{theorem}\label{theom:odd vector field}
A  weighted skew  algebroid of degree $k$ can be defined as an odd vector field of weight $(0,1)$ on the double graded super bundle $\Pi D_{k}$. Moreover, if we deal with a weighted Lie  algebroid, then the odd vector field is a homological vector field.
\end{theorem}

\begin{proof}
Let us equip the supermanifold $\Pi \sT^{*}\Pi D_{k}$ with local coordinates

\begin{equation*}
 (\underbrace{x^{\alpha}_{u}}_{(u,0,0)}, ~ \underbrace{\theta^{I}_{u+1}}_{(u,1,0)}, ~  \underbrace{\chi_{\beta}^{u+2}}_{(u,1,1)}, ~ \underbrace{\eta_{J}^{u+1}}_{(u,0,1)}  ),
\end{equation*}
where we have  of course anticipated the form of the  canonical symplectomorphism, which is given by

\begin{equation*}
R^{*}(\pi_{I}^{u+1}) = - \eta_{I}^{u+1}.
\end{equation*}
Then applying $R$ to $\sP_{\epsilon}$ and then undoing the symbol map and noting the change in Grassmann parity and the shift in the weight by $(1-k,-1,-1)$, we arrive at the following  vector field
\begin{equation*}
Q_{\epsilon} = \theta^{I}_{u-u'+1}P^{\alpha}_{I}[u'](x)\frac{\partial}{\partial x^{\alpha}_{u}} - \frac{1}{2!}\theta^{J}_{u-u'+1} \theta^{I}_{u''-u' +1}P_{IJ}^{K}[u'](x)\frac{\partial}{\partial \theta^{K}_{u+1}},
\end{equation*}

\noindent on $\Pi D_{k}$ which is of  bi-weight $(0,1)$, meaning it is \emph{Grassmann odd} and of \emph{total weight one}. If we have a weighted  Lie algebroid, then the associated bi-vector field is a Poisson structure and the odd vector field above is homological, that is $[Q_{\epsilon}, Q_{\epsilon}]=0$ \emph{via} our comments opening this subsection.
\end{proof}

\begin{corollary}\label{cor}
A weighted skew (Lie) algebroid can equivalently be defined as a skew (Lie) algebroid $E$ equipped with a homogeneity structure $h: \mathbb{R} \times E \rightarrow E$ that acts as skew (Lie) algebroid morphisms for all $t \in \mathbb{R}$.
\end{corollary}

\begin{remark} It is worth further pointing out the relation between weighted Lie algebroids and Voronov's non-linear Lie algebroids. Voronov \cite{voronov-2010} takes as his \emph{definition} of a non-linear Lie algebroid the pair $(\mathcal{M}, ~ Q)$, where $\mathcal{M}$ is a non-negatively graded supermanifold  and $Q$ is a homological vector field of weight $1$. Note that we must have the structure of a supermanifold here in order to define a non-trivial odd vector field.  Generally, there is no (graded) manifold $N$ such that  `$ \mathcal{M} \simeq  \Pi N$'; there is no obvious parity reversion functor here. Thus, in this respect Voronov's general definition of a non-linear Lie algebroid requires the underlying structure of a supermanifold. Weighted Lie algebroids as we initially defined them requires no initial supermanifold structure. As we have the underlying structure of a $\mathcal{GL}$-bundle  one can, as we have done, take duals and apply the standard parity reversion functor in order to build a supermanifold and a homological vector field. Then by passing to total weight we see that weighted Lie algebroids form a special class of non-linear Lie algebroids that is suitable for applications in pure even differential geometry and geometric mechanics. In essence because we have a linear structure the theory of weighted Lie algebroids is closer to the theory of standard Lie algebroids than Voronov's non-linear Lie algebroids.
\end{remark}

\begin{proposition}\label{prop:A1algebroid}
If $(D_{k}, \epsilon)$ is  a  weighted skew/Lie algebroid, then $A_{1} \rightarrow M$ is a skew/Lie algebroid.
\end{proposition}
\begin{proof}
The above proposition follows from the fact that $Q_{\epsilon} \in \Vect(\Pi D_{k})$ projects to a weight one vector field on $\Pi A_{1}$.
\end{proof}
If we denote the projection of $Q_{\epsilon}$ to $\Pi A_{1}$ as $\D_{\epsilon}$ then in homogeneous local coordinates it is of the form

\begin{equation}
\D_{\epsilon} = \theta_{1}^{\nu}P_{\nu}^{A}(x) \frac{\partial }{\partial x^{A}} - \frac{1}{2!} \theta_{1}^{\nu}\theta_{1}^{\mu}P_{\mu \nu}^{\rho}\frac{\partial}{\partial \theta_{1}^{\rho}},
\end{equation}

\noindent which is clearly of the form expected. Furthermore, the odd vector field that encodes a weighted skew algebroid structure satisfies the Leibniz rule

\begin{equation}
Q_{\epsilon}(\alpha \: \Phi) = \D_{\epsilon}(\alpha) \: \Phi  + (-1)^{\widetilde{\alpha}} \alpha \:\: Q_{\epsilon}(\Phi),
\end{equation}

\noindent for $\alpha \in C^{\infty}(\Pi A_{1}) \subset C^{\infty}(\Pi D_{k})$  and $\Phi \in C^{\infty}(\Pi D_{k})$, where $\widetilde{\alpha}$ is the Grassmann parity of $\alpha$, which in this case corresponds to the form degree. One should immediately be reminded of the concept of a Lie algebroid module  as defined by Va$\breve{\textrm{{\i}}}$ntrob \cite{Vaintrob:1997}. However, we have a  non-linear version of this concept as $\Pi D_{k} \rightarrow \Pi A_{1}$ is a general  fibration and not necessarily a linear fibration.

\begin{example}\label{exm:Tangentalgebroid}
The tangent bundle of a graded bundle  $F_{k-1}$ of degree $k-1$ is canonically a weighted Lie algebroid of degree $k$.
Let us employ homogeneous local coordinates $(x^{I}_{u}, \zx_{u+1}^{J})$, where $1\leq u \leq k-1$ on $\Pi \sT F_{k-1}$;  the tangent bundle is clearly a $\mathcal{GL}$-bundle. The canonical de Rham differential, which is a homological vector field
\begin{equation*}
Q = \zx^{J}_{u+1} \frac{\partial}{\partial x^{J}_{u}},
\end{equation*}
\noindent is of bi-weight $(0,1)$, and thus describes a weighted Lie algebroid structure on $\sT F_{k-1}$ \emph{via} Theorem \ref{theom:odd vector field}.
\end{example}

\begin{proposition}\label{prop:Cotangentalgebroid}
If $(F_{k-1}, \mathcal{P})$ is a Poisson manifold of degree $k-1$, i.e. $\mathcal{P} \in C^{\infty}(\Pi \sT^{*}F_{k-1})$ such that it is is of bi-weight $(k-1,2)$ and $\SN{\mathcal{P},\mathcal{P}}=0$, where the bracket here is the canonical Schouten bracket on $\Pi \sT^{*}F_{k-1}$,  then  $\sT^{*}F_{k-1}$ comes with the structure of a weighted Lie algebroid. If we consider an almost Poisson manifold of degree $k-1$, i.e. we loose the vanishing of the Schouten bracket, then  $\sT^{*}F_{k-1}$ comes with the structure of a weighted skew algebroid. Furthermore, the vector bundle $\bar{F}_{k-1}^{*} \rightarrow M$ is a skew/Lie algebroid.
\end{proposition}

\begin{proof}
Clearly $\sT^{*}F_{k-1}$ is a $\mathcal{GL}$-bundle. Furthermore, the odd vector field $Q = - \SN{\mathcal{P}, \bullet} \in \Vect(\Pi \sT^{*}F_{k-1})$ is of bi-weight $(0,1)$, as $\mathcal{P}$ is of bi-weight $(k-1,2)$ by assumption and the canonical Schouten bracket is of bi-weight $(1-k,-1)$. This odd vector field thus, in light of Theorem \ref{theom:odd vector field}, encodes a weighted skew algebroid structure. If $\mathcal{P}$ is Poisson, then the odd vector field is clearly homological. Proposition \ref{prop:A1algebroid} then establish that  $\bar{F}_{k-1}^{*} \rightarrow M$ is a skew/Lie algebroid.
\end{proof}

The above proposition provides us with the natural generalisation of the cotangent Lie algebroid of a classical Poisson manifold. Indeed, if we restrict ourselves to the case of $k=2$, then we are dealing precisely with the cotangent Lie algebroid of a classical Poisson manifold.

 It is easy to see that, for a vector bundle $\zt:E\to M$, the higher tangent bundle $\sT^{k-1}E$ is a $\mathcal{GL}$-bundle of degree $k$, with the vector bundle structure $\sT^{k-1}\zt:\sT^{k-1}E\to\sT^{k-1}M$.

\begin{proposition}\label{prop:liftalgebroid}
If $E \rightarrow M$ is a Lie algebroid, then $\sT^{k-1}E$ is  canonically a weighted Lie algebroid of degree $k$.
\end{proposition}

\begin{proof}
We have the isomorphism $\Pi (\sT^{k-1}E) \simeq \sT^{k-1}\Pi E$ which can be directly verified using local coordinates: for a geometric definition of the higher order tangent bundle of a supermanifold see \cite{Bruce:2014}. The complete lift of the homological vector field on $\Pi E$ that encodes the Lie algebroid structure to the $k-1$ th order tangent bundle encodes the weighted Lie algebroid structure. For the complete lift of vector fields to higher order tangent bundles  see for example \cite{Kolar:1988} and references therein.
\end{proof}

\subsection{Relation to $\mathcal{VB}$-algebroids}
In view of theorem \ref{theom:odd vector field} and corollary \ref{cor}, weighted Lie algebroids are generalizations of \emph{$\mathcal{VB}$-algebroids} as defined by Gracia-Saz \& Mehta  \cite{Gracia-Saz:2009} and, alternatively, by Bursztyn, Cabrera \& del Hoyo \cite{Bursztyn:2014}, or \emph{$\mathcal{LA}$-vector bundles}, as defined by Mackenzie  \cite{Mackenie:1998}.

Recall that Gracia-Saz \& Mehta  \cite{Gracia-Saz:2009} define  a $\mathcal{VB}$-algebroid as  a homological vector of weight $(0,1)$ (note they use the opposite ordering of the weights) on the double (super) vector bundle\\
\vspace{10pt}
\begin{tabular}{p{5cm} p{10cm}}
\leavevmode
\begin{center}
\begin{xy}
(0,15)*+{\Pi_{E} D}="a"; (25,15)*+{\Pi A}="b";%
(0,0)*+{E}="c";  (25,0)*+{M}="d";%
{\ar "a";"b"}; {\ar "a";"c"}; %
{\ar "b";"d"}; {\ar "c";"d"}; %
\end{xy}
\end{center}
&
\vspace{10pt}
The situation for weighted  skew algebroids is very similar, but with some subtle differences. Principally we no longer have a double vector bundle structure, but a slightly more general double graded bundle structure.
\end{tabular}

\vspace{-15pt}
\begin{tabular}{p{10cm} p{5cm}}
In particular, the corresponding diagram consists of vertical arrows which are genuine vector bundle structures while the horizontal arrows are more general  fibrations.
&
\vspace{-15pt}
\leavevmode
\begin{center}
\begin{xy}
(0,15)*+{\Pi D_{k}}="a"; (25,15)*+{\Pi A_{1}}="b";%
(0,0)*+{B_{k-1}}="c";  (25,0)*+{M}="d";%
{\ar "a";"b"}; {\ar "a";"c"}; %
{\ar "b";"d"}; {\ar "c";"d"}; %
\end{xy}
\end{center}
\end{tabular}

\noindent where we have defined $A_{1} :=  D_{k}[\Delta^{1} \leq 0]$, which is a vector bundle over $M$. As before, $B_{k-1}:=  D_{k}[\Delta^{2} \leq 0]$.

Gracia-Saz \& Mehta  \cite{Gracia-Saz:2009} showed that $\mathcal{VB}$-algebroids are equivalent to $\mathcal{LA}$-vector bundles, as defined by Mackenzie  \cite{Mackenie:1998} as a double vector bundle for which the horizontal/vertical sides are Lie algebroids and the structure maps of the vertical/horizontal vector bundles are Lie algebroid morphisms.

Alternatively, Bursztyn, Cabrera \& del Hoyo proved in \cite{Bursztyn:2014} that a $\mathcal{VB}$-algebroid can equivalently be defined as a Lie algebroid equipped with a regular homogeneity structure (vector bundle structure) that acts (via homotheties) as Lie algebroid morphisms for all $t \in \mathbb{R}$.

\begin{proposition}\label{prop:VBalgebroids}
Weighted Lie algebroids of degree two are precisely   $\mathcal{VB}$-algebroids.
\end{proposition}

\begin{proof}
It is clear from the definition that a weighted  Lie algebroid of degree two is  a $\mathcal{VB}$-algebroid  via Theorem \ref{theom:odd vector field}  and Theorem 3.16 of \cite{Gracia-Saz:2009}, which established the description of $\mathcal{VB}$-algebroids in terms of homological vector fields of weight $(0,1)$. The converse also follows from the same theorems.
\end{proof}

\subsection{Weighted skew/Lie algebras}
\begin{definition}
A weighted skew/Lie algebroid, $(D_{k}, \epsilon),$ is called a \emph{weighted skew/Lie algebra} if
\begin{equation*}
D_{k}[\Delta^{1}+ \Delta^{2} \leq 0] = \{ \textnormal{pt}\}.
\end{equation*}
\end{definition}

Interestingly, this algebraic counterpart of a weighted skew/Lie algebroid still has a rich graded structure, we do not simply recover a graded Lie algebra. In particular we still have a double structure

\begin{tabular}{p{5cm} p{10cm}}
\leavevmode
\begin{center}
\begin{xy}
(0,15)*+{\Pi D_{k}}="a"; (25,15)*+{\Pi A_{1}}="b";%
(0,0)*+{B_{k-1}}="c";  (25,0)*+{\{\textnormal{pt}\}}="d";%
{\ar "a";"b"}; {\ar "a";"c"}; %
{\ar "b";"d"}; {\ar "c";"d"}; %
\end{xy}
\end{center}
&
\vspace{10pt}
 Because we consider the homogeneity space $B_{k-1}$ as a manifold, we still have a skew/Lie algebroid structure on $D_{k} \rightarrow B_{k-1}$, but now $A_{1}$ is a genuine skew/Lie algebra.
 \end{tabular}

\begin{example}
To illustrate the general situation, consider a degree 2  weighted skew/Lie algebra. Let us on $\Pi D_{2}$ employ homogeneous coordinates
\begin{equation*}
(\underbrace{y^{a}_{1}}_{(1,0)}, ~ \underbrace{\zx_{1}^{\nu}}_{(0,1)}, \underbrace{\theta^{i}_{2}}_{(1,1)}).
\end{equation*}
Then, the most general bi-weight $(0,1)$ Grassmann odd vector field on   $\Pi D_{2}$ must be of the form
\begin{equation*}
Q = (\zx_{1}^{\nu} y_{1}^{a}P_{a \nu}^{b} + \theta_{2}^{j}P_{j}^{b})\frac{\partial}{\partial y^{b}_{1}} - \frac{1}{2} \zx_{1}^{\mu}\zx_{1}^{\nu}P_{\nu \mu }^{\rho} \frac{\partial}{\partial \zx_{1}^{\rho}} - (\theta_{2}^{j}\zx^{\nu}_{1}P_{\nu j}^{i} + \frac{1}{2}\zx^{\nu}_{1} \zx^{\mu}_{1}y^{a}_{1}P_{a\mu \nu}^{i})\frac{\partial}{\partial \theta^{i}_{2}}.
\end{equation*}
Just by inspection we see that there is skew/Lie algebroid structure $D_{2} \rightarrow B_{1}$ and a skew/Lie algebra structure on $A_{1}$.
\end{example}

\begin{proposition}
If ${\mathfrak{g}}$ is a Lie algebra and $V$ is an arbitrary vector space, then ${\sT} {\mathfrak{g}} \oplus V$  is canonically a weighted algebra of degree 2.
\end{proposition}

\begin{proof}
First, let us equip $ \sT \mathfrak{g} \oplus V$ with homogeneous coordinates
\begin{equation*}
(\underbrace{u^{a}}_{(0,1)}, ~ \underbrace{\dot{u}^{b}}_{(1,1)}, ~\underbrace{{v}^{i}}_{(1,0)}),
\end{equation*}
where the ordering of the assignment of bi-weight is set by our earlier conventions. Then it is not hard to see that $\sT \mathfrak{g} \oplus V$ is a $\mathcal{GL}$-bundle, in fact it is a double vector bundle. Now, let us shift the parity, which is defined by the second entry of the bi-weight, and consider

\begin{tabular}{p{5cm} p{10cm}}
\leavevmode
\begin{center}
\begin{xy}
(0,15)*+{\sT \Pi \mathfrak{g} \oplus V}="a"; (25,15)*+{\Pi \mathfrak{g}}="b";%
(0,0)*+{V}="c";  (25,0)*+{\{ \textnormal{pt}\}}="d";%
{\ar "a";"b"}; {\ar "a";"c"}; %
{\ar "b";"d"}; {\ar "c";"d"}; %
\end{xy}
\end{center}
&
\vspace{10pt}
 As $\mathfrak{g}$ is a Lie algebra, there exists a homological vector field $\rm{d}_{\mathfrak{g}} \in \Vect(\Pi \mathfrak{g})$ of weight one. In line with our assignment of the bi-weight this homological vector field is $(0,1)$.
 \end{tabular}

The homological vector field encoding the weighted Lie algebra structure is then just the complete lift of $\rm{d}_{\mathfrak{g}}$ to  a vector field on $\sT \Pi \mathfrak{g}$ which is naturally included in $\Vect(\sT \Pi \mathfrak{g} \oplus V)$.
 \end{proof}
We will meet a less trivial example of a weighted Lie algebra in Section \ref{sec:examples}.

\subsection{Weighted algebroid structures on the linearisation of a graded bundle}
Our initial motivation for defining weighted algebroids was to find some alternative notion of a ``higher Lie algebroid" on a graded bundle. To do this we place the further condition on the underlying $\mathcal{GL}$-bundle structure to be the linearisation of a graded bundle. Although this may seem a very heavy restriction natural examples are still plentiful.

\begin{definition} We will say that a graded bundle $F_{k}$ \emph{carries the structure of a weighted algebroid} if and only if  there exists a graded map $\epsilon: \sT^{*}D(F_{k}) \rightarrow \sT D^{*}(F_{k})$, such that $(D(F_{k}), \epsilon)$ is a weighted algebroid. Furthermore, we will also say that $\epsilon$ is a \emph{weighted algebroid for the graded bundle} $F_{k}$ if and only if  $(D(F_{k}), \epsilon)$ is a weighted algebroid.
\end{definition}

We also have a little more freedom in how we understand the anchor map.

\begin{definition}
Let $\epsilon$ be a weighted algebroid for $F_{k}$, then we define the series of \emph{graded bundle anchors} as
\begin{equation*}
\hat\rho_{q} : F_{k} \rightarrow \sT F_{q-1},  \hspace{30pt} (1\leq q \leq k)
\end{equation*}
where $\hat\rho_{q} :=  \sT \tau^{k-1}_{q-1} \circ \rho \circ \iota_{F_k}$ and for consistency we take $\tau^{k-1}_{k-1} := id_{F_{k-1}}$.
\end{definition}

In the two extremes we have $\hat\rho_{k}: F_{k} \rightarrow \sT F_{k-1}$ and $\hat\rho_{1}: F_{k} \rightarrow \sT M$. Thus, the first order anchor mimics the standard anchor of a Lie algebroid very closely.

\begin{proposition}\label{prop:F1algebroid}
If the graded bundle $F_{k}$ carries the structure of a  weighted skew/Lie algebroid, then $F_{1}\rightarrow M$ has the structure of a skew/Lie algebroid.
\end{proposition}
\begin{proof}
This follows directly from Proposition \ref{prop:A1algebroid} and the fact that for the linearisation of a graded bundle  $B_{k-1} \simeq F_{k-1}$ and $A_{1} \simeq F_{1}$.
\end{proof}

One would expect a higher order tangent bundle of manifold to be a canonical example of a  `higher Lie algebroid',  given than the tangent bundle of a manifold is naturally a Lie algebroid. Indeed via the linearisation we arrive at the following theorem.
\begin{theorem}\label{theom:highertangent}
The k-th order tangent bundle $\sT^{k}M$ of a manifold $M$ carries a canonical weighted Lie algebroid structure.
\end{theorem}
\begin{proof}
The linearisation of the k-th order tangent bundle is  $D(\sT^{k}M) \simeq \sT (\sT^{k-1}M)$ and we  identify the homological vector field encoding the weighted  Lie algebroid as the canonical de Rham differential on $\Pi \sT (\sT^{k-1}M)$.
\end{proof}

\begin{remark}
Recall that the tangent bundle and canonical Lie bracket of vector fields can be thought of as the Lie algebroid associated with the pair groupoid (or the fundamental groupoid). One takes a reduction of the tangent bundle of the pair groupoid, which gives us the tangent bundle on one of the factors of the pair groupoid. The weighted Lie algebroid structure on the linearisation of a k-th order tangent bundle should be viewed in a similar light. That is, as the linearisation of the reduction of the k-th order tangent bundle of the pair groupoid.  We leave details to the reader and will consider other examples in Section \ref{sec:examples}.
\end{remark}

\begin{example}
Let us consider  $F_{2}$, which we equip with local homogeneous coordinates $(x^{A}, y^{a} , z^{i})$ of weight $0,1$ and $2$ respectively. The graded bundle $F_{2}$ carries the structure of a weighted Lie algebroid if $(D(F_{2}), \epsilon)$ is a weighted Lie algebroid. Let us then equip $\Pi D(F_{2})$ with natural homogeneous coordinates

\begin{equation*}
(\underbrace{x^{A}}_{(0,0)}, ~ \underbrace{y^{a}}_{(1,0)}, ~ \underbrace{\zx^{b}}_{(0,1)}, ~ \underbrace{\theta^{i}}_{(1,1)}),
\end{equation*}

\noindent  one should note that we have a double (super) vector bundle. The homological vector field that encodes the weighted Lie algebroid structure is of the local form:

\begin{equation*}
Q_{\epsilon} = \zx^{a}P_{a}^{A}\frac{\partial}{\partial x^{A}}- \frac{1}{2!} \zx^{a} \zx^{b}P_{ba}^{c} \frac{\partial}{\partial \zx^{c}} + (\zx^{a}y^{b} \bar{P}_{ba}^{c} + \theta^{i}P_{i}^{c}) \frac{\partial}{\partial y^{c}} -\frac{1}{2!}( 2 \theta^{j}\zx^{a}P_{aj}^{i} + \zx^{a} \zx^{b}y^{c}P_{cba}^{i}) \frac{\partial}{\partial \theta^{i}} \in \Vect(\Pi D(F_{2})).
\end{equation*}
\noindent It is easy to see that the first two terms of this homological vector field correspond to a Lie algebroid structure on $F_{1} \rightarrow M$. Furthermore it is clear that we have a pair of anchors $\rho_{2}: F_{2} \rightarrow \sT F_{1}$ and $\rho_{1}: F_{2} \rightarrow \sT M$. It is sufficient to describe only the anchor $\rho_{2}$ in any detail as the anchor $\rho_{1}$ follows via natural projection. In obvious coordinates on $\sT F_{1}$ we have
\begin{equation*}
(\rho_{2})^{*}(\delta x^{A}) = y^{a}P_{a}^{A}(x), \hspace{25pt} (\rho_{2})^{*}(\delta y^{a}) = 2 z^{i}P_{i}^{a}(x)  + y^{b}y^{c}\bar{P}_{cb}^{a}(x).
\end{equation*}
\end{example}

\section{Further examples of weighted Lie algebroids}\label{sec:examples}
\subsection{Recap of given examples}
In the previous section several examples of weighted Lie algebroids were given, let us just list these examples again for convenience:
\begin{enumerate}
\item $\mathcal{VB}$-algebroids, and so $\mathcal{LA}$-vector bundles.
\item The cotangent bundle of a graded bundles of degree $k-1$ equipped with a Poisson structure of weight $k-1$, this includes the standard cotangent Lie algebroid of a classical Possion manifold.
\item Tangent bundles of a graded bundle, which  of course includes the tangent bundle of a vector bundle.
\item Higher order tangent bundles of the total space of a Lie algebroid.
\item The linearisation of a higher order tangent bundle.
\item The Atiyah algebroid of a graded principal bundle.
\end{enumerate}

\noindent Thus, weighted algebroids represent a unification of several structures found in the theory of Lie algebroids and graded bundles. All the examples met so far are a natural mix of examples of graded bundles and Lie algebroids: by thinking of weighted Lie algebroids as algebroid objects in the category of graded bundles the examples are quite obvious, but none the less deserve stating clearly. More interesting and potentially useful  examples can be found via reductions of higher tangent bundles.

\subsection{Reductions of  higher order tangent bundles of Lie groups}\label{group-red}
There are two extreme examples of Lie algebroids. The first one is the canonical algebroid on a tangent bundle $\tau_M: \sT M\rightarrow M$ with the bracket of vector fields and the anchor being the identity map. The second example is the Lie algebra $\mathfrak{g}$ of a Lie group $G$ regarded as a vector bundle over one point manifold with the anchor being the zero map. In the first example the corresponding double vector bundle morphism $\ze_{M}: \sT^\ast\sT M\rightarrow \sT\sT^\ast M$ is the inverse of the Tulczyjew morphism $\alpha_{M}$. In the second example $\ze_\mathfrak{g}$ reads

$$\ze_{\mathfrak{g}}: \mathfrak{g}\times\mathfrak{g}^\ast\ni(X,\zx)\longmapsto (\zx, ad_X^\ast\zx)\in  \mathfrak{g}^\ast\times\mathfrak{g}^\ast. $$
The map $\ze_{\mathfrak{g}}$ can be obtained from $\ze_{G}$ by reduction with respect to the group action by left translations. It corresponds of course to the Lie algebra bracket $[\cdot,\cdot]_\mathfrak{g}$.
Let us examine this example in more details. Tangent and cotangent bundles of a Lie group $G$ can be trivialized using lifted action of the group on itself by left translations. It means that we have $\sT G\simeq G\times \mathfrak{g}$ and $\sT^\ast G\simeq G\times \mathfrak{g^\ast}$. Iterated tangent and cotangent bundles can be consequently trivialised: $\sT^\ast\sT G\simeq G\times \mathfrak{g}\times \mathfrak{g}^\ast\times \mathfrak{g}^\ast$ and $\sT\sT^\ast G\simeq G\times \mathfrak{g}^\ast\times \mathfrak{g}\times \mathfrak{g}^\ast$. In these trivialisations the canonical map $\ze_{G}$ reads

$$\ze_G: G\times \mathfrak{g}\times \mathfrak{g}^\ast\times \mathfrak{g}^\ast\ni(g, X, \zf, \psi)\longmapsto (g, \psi, X, \zf+ad_X^\ast\psi)\in G\times \mathfrak{g}^\ast\times \mathfrak{g}\times \mathfrak{g}^\ast.$$
There exists a coisotropic submanifold $K\in\sT^\ast\sT G$ elements of which are differentials of $G$-invariant functions on $\sT G$. In other words, $K$ is the annihilator of the distribution on $\sT G$ spanned by fundamental vector fields of the $G$-action on $\sT G$ lifted from the regular action on $G$. In the trivialisation as above, $K=\{(g, X, \zf, \psi):\; \zf=0\}$. The target space of symplectic reduction of $\sT^\ast\sT G$ with respect to $K$ can be identified with $\sT^\ast \mathfrak{g}\simeq \mathfrak{g}\times\mathfrak{g}^\ast$. The reduction will be denoted by $\rho$. On the other hand, according to the trivialisation $\sT^\ast G\simeq G\times \mathfrak{g}^\ast$ there exists the projection $p_2:\sT^\ast G\rightarrow \mathfrak{g}^\ast$. The $\ze_\mathfrak{g}$ morphism can be obtained as

$$\ze_\mathfrak{g}=\sT p_2\circ\ze_{G}\circ\rho^{-1}.$$
The same method can be used to derive the canonical weighted Lie algebra structure on the linearisation of the graded bundle $\mathfrak{g}_2=\sT^2 G\slash G$ over one-point manifold. Since $\mathfrak{g}_2$ is a reduction of the graded bundle $\sT^2 G\rightarrow G$, we expect that the weighted Lie algebra structure on the linearisation of $\mathfrak{g}_2$ is a reduction of the canonical algebroid structure on the linearisation of $\sT^2G$, i.e. $\sT\sT G$. The canonical algebroid structure on $\sT\sT G$ is the morphism $\ze_{\sT G}$, the inverse of Tulczyjew morphism $\alpha_{\sT G}$, i.e. the map
$$\ze_{\sT G}:\;\sT^\ast\sT\sT G\longrightarrow \sT\sT^\ast\sT G\,,$$
which in the trivialisations reads
\begin{align*} G\times \mathfrak{g}\times  \mathfrak{g}\times  \mathfrak{g}\times  \mathfrak{g}^\ast\times  \mathfrak{g}^\ast\times  \mathfrak{g}^\ast\times \mathfrak{g}^\ast&\longrightarrow
\ze_{\sT G}: G\times \mathfrak{g}\times  \mathfrak{g}^\ast\times  \mathfrak{g}^\ast\times  \mathfrak{g}\times  \mathfrak{g}\times  \mathfrak{g}^\ast\times \mathfrak{g}^\ast\,, \\
 (g,X,Y,Z, a,b,c,d)&\longmapsto (g,X,c,d,Y,Z, a+ad_Y^\ast c, b)\,.\end{align*}
Using the lifted group action  and taking care on the weights, we can identify $\sT^2G$ with $G\times \mathfrak{g}[1]\times \mathfrak{g}[2]$, and $\mathfrak{g}_2$ with $\mathfrak{g}[1]\times \mathfrak{g}[2]$. Then, $D(\sT^2 G)$ is identified with $G\times \mathfrak{g}[(1,0)]\times  \mathfrak{g}[(0,1)]\times  \mathfrak{g}[(1,1)]$ and  $D(\mathfrak{g}_2)$ with $\mathfrak{g}[(1,0)]\times \mathfrak{g}[(0,1)]\times \mathfrak{g}[(1,1)]$. The linearisation $D(\mathfrak{g}_2)$ is a vector bundle over the first factor. The reduction procedure, as before, consists of a symplectic reduction of $\sT^\ast\sT\sT G$, which in the trivialisation means putting the fifth factor equal to zero and omitting the first factor, and the appropriate projection from $\sT\sT^\ast\sT G$ (omitting the first and fifth factors). In the trivialisation the reduced $\ze_2$ reads

\begin{align}\ze_2: \sT^\ast D(\mathfrak{g}_2)\simeq\mathfrak{g}\times  \mathfrak{g}\times  \mathfrak{g}\times  \mathfrak{g}^\ast\times  \mathfrak{g}^\ast\times  \mathfrak{g}^\ast&\longrightarrow
 \mathfrak{g}\times  \mathfrak{g}^\ast\times  \mathfrak{g}^\ast\times  \mathfrak{g}\times  \mathfrak{g}^\ast\times  \mathfrak{g}^\ast\simeq \sT D^\ast(\mathfrak{g}_2)\,, \\
(X,Y,Z,b,c,d)& \longmapsto(X,c,d,Z, ad_Y^\ast c, b)\,.\end{align}
Here, $(X,Y,Z,b,c,d)$ have weights $(1,0,0),(0,1,0),(1,1,0),(0,1,1),(1,0,1),(0,0,1)$, respectively.
The above morphism corresponds to the following bracket of sections of the bundle $D(\mathfrak{g}_2)\rightarrow\mathfrak{g}$:
\begin{equation}\label{reduced-bracket}[(Y_1, Z_1),(Y_2, Z_2)]_2=([Y_1, Y_2]_\mathfrak{g}+Z_1(Y_2)-Z_2(Y_1), \, Z_1(Z_2)-Z_2(Z_1)).\end{equation}
Functions $Z_i(X)$ and $Y_i(X)$ can be treated as both, vector fields on $\mathfrak{g}$ and functions on $\mathfrak{g}$ with values in $\mathfrak{g}$. In the notation $Z_i(Y_j)$, the symbol $Z_i$ denotes the vector field $Z_i$ acting on the function $Y_j$. One can see that the anchor map for $\ze_2$ is
\begin{equation}\rho_2: D(\mathfrak{g}_2)\ni (X,Y,Z)\longmapsto (X,Z)\in \mathfrak{g}[(1,0)]\times\mathfrak{g}[(1,1)]\simeq \sT \mathfrak{g}[1].\end{equation}
 The second component in $[\cdot,\cdot]_2$ is then just the bracket of vector fields $Z_1$ and $Z_2$.

The analogous construction can be done for $\sT^kG$ and $\mathfrak{g}_k=\sT^k G\slash G$.
Let us remark that what we get is exactly the Atiayh algebroid of the graded principal bundle $\sT^{k-1}G\to \sT^{k-1}G/G=\mathfrak{g}_{k-1}$ with the structure group $G$ (see example \ref{Atiyah}).
Note, however, that as the Lie algebra $\mathfrak{g}$ is usually interpreted in terms of left-invariant vector fields, our principal bundles here will be equipped with the left, instead of the right action of $G$. Then, $G$ acts only on the first factor of the left-trivialization $G\ti\mathfrak{g}_{k}\simeq\sT^kG$.

The sections of the Atiyah algebroid $\sT\sT^{k-1}G/G\to\mathfrak{g}_{k-1}$ are $G$-invariant vector fields on $$\sT^{k-1}G=G\ti\mathfrak{g}[1]\ti\cdots\ti\mathfrak{g}[k-1]=G\ti\mathfrak{g}_{k-1}\,,$$
which can be identified with sections of the (trivial) vector bundle
$$(\mathfrak{g}_{k-1}\ti\mathfrak{g})\ti_{\mathfrak{g}_{k-1}}\sT\mathfrak{g}_{k-1}\to\mathfrak{g}_{k-1}\,,$$
i.e. functions $$(Y,Z):\mathfrak{g}_{k-1}\to \mathfrak{g}\ti\mathfrak{g}_{k-1}\,.$$
Moreover, $Y(X)$'s commute as invariant vector fields on $G$, i.e. elements of the Lie algebra $\mathfrak{g}$, and $Z(X)$'s as vector fields on $\mathfrak{g}_{k-1}$.  This gives us the formula  the for the Lie algebroid bracket formally like (\ref{reduced-bracket}):
\begin{equation}\label{reduced-bracket1}[(Y_1, Z_1),(Y_2, Z_2)]_k=([Y_1, Y_2]_\mathfrak{g}+Z_1(Y_2)-Z_2(Y_1), \, Z_1(Z_2)-Z_2(Z_1))\,,\end{equation}
with the only difference that now $Y(X)\in\mathfrak{g}$, $Z(X)\in \mathfrak{g}_{k-1}$, and $X\in \mathfrak{g}_{k-1}$.

\medskip

\noindent \textbf{As a homological vector field.} From the above discussion it is clear that $\Pi D(\sT^{2}G\slash G) = \mathfrak{g}[(1,0)] \times \Pi \mathfrak{g}[(0,1)] \times \Pi \mathfrak{g}[(1,1)]$. Let us equip this supermanifold with local coordinates
 \begin{equation*}
 (\underbrace{y^{a}}_{(1,0)}, ~ \underbrace{\zx^{b}}_{(0,1)}, ~ \underbrace{dy^{c}}_{(1,1)}),
 \end{equation*}
\noindent where we have correctly interpreted the third collection of coordinates as the differentials of the first set of coordinates. Then, \emph{via} inspection we see that the homological vector field encoding the weighted Lie algebra structure on the linearisation of $\sT^{2}G\slash G$ is given by
\begin{equation}
Q = dy^{a}\frac{\partial}{\partial y^{a}} - \frac{1}{2}\zx^{a} \zx^{b}P_{ba}^{c} \frac{\partial}{\partial \zx^{c}},
\end{equation}
\noindent where $P_{ba}^{c}$ is (up to a sign) the structure constant of the Lie algebra $(\mathfrak{g}, [\cdot, \cdot]_{\mathfrak{g}})$. Generalising this to higher tangent bundles of Lie groups is clear. To see this, consider
$$\Pi \mathfrak{g}_{k} = \Pi \sT^{k}G \slash G = \displaystyle \times_{i=1}^{k-1} \mathfrak{g}[(i,0)] \times \Pi \mathfrak{g}[(0,1)] \times_{j=1}^{k-1} \Pi \mathfrak{g}[(j,1)]$$ which we equip with homogeneous local coordinates

\begin{equation*}
(\underbrace{y^{a}_{r}}_{(r,0)}, ~ \underbrace{\zx^{b}}_{(0,1)}, ~ \underbrace{dy^{c}_{r+1}}_{(r,1)}),
\end{equation*}
\noindent where $1 \leq r < k$. The homological vector field encoding the weighted Lie algebra on $\mathfrak{g}_{k}$ is given by

\begin{equation}
Q = \sum_{r=1}^{k-1}dy^{a}_{r+1}\frac{\partial}{\partial y^{a}_{r}} -  \frac{1}{2}\zx^{a} \zx^{b}P_{ba}^{c} \frac{\partial}{\partial \zx^{c}}.
\end{equation}

The above homological vector field can then be interpreted as the de Rham differential on $\Pi \sT \left(\sT^{k-1}G\slash G\right)$ plus the standard Chevalley--Eilenberg differential on $\Pi \mathfrak{g}$.  The form of the homological vector field is somewhat expected, as we are dealing with an Atiayh algebroid.

\subsection{Reductions of higher order tangent bundles of Lie groupoids}
Similar to the case of a Lie group, reductions of higher order tangent bundles of Lie groupoids give rise to a class of weighted Lie algebroids. There are of course technicalities that do not arise in the group case that require careful handling in the groupoid case. For an introduction to the theory of Lie groupoids the reader can consult Mackenzie \cite{Mackenzie2005}. We will also freely use some of the constructions of J\'{o}\'{z}wikowski \&  Rotkiewicz \cite{Jozwikowski:2014,Jozwikowski:2015}.

Let $\mathcal{G} \rightrightarrows M$ be an arbitrary Lie groupoid with source map $\underline{s}: \mathcal{G} \rightarrow M$ and target map $\underline{t}: \mathcal{G} \rightarrow M$. There is also the inclusion map $\iota : M \rightarrow \mathcal{G}$   and a partial multiplication  $(g,h) \mapsto g\cdot h$ which is defined on $\mathcal{G}^{(2)}=\{(g,h)\in\mathcal{G}\ti\mathcal{G}: \underline{s}(g) = \underline{t}(h)\}$. Moreover, the manifold $\mathcal{G}$ is foliated by $\underline{s}$-fibres $\mathcal{G}_{x}= \{ \left.g \in \mathcal{G}\right| \underline{s}(g) =x\}$, where $x \in M$. As by definition the source (and target) map is a submersion, the $\underline{s}$-fibres are themselves smooth manifolds. Geometric objects associated with this foliation will be given the superscript $\underline{s}$. For example, the distribution tangent to the leaves of the foliation will be denoted by $\sT \mathcal{G}^{\underline{s}}$. We will need to consider a higher order version of this, that is the subbundle $\sT^{k}\mathcal{G}^{\underline{s}} \subset \sT^{k}\mathcal{G}$ consisting of all higher order velocities tangent to some $\underline{s}$-leaf $\mathcal{G}_{x}$. That is we identify $\sT^{k}\mathcal{G}^{\underline{s}} $ with the union of $\sT^{k}\mathcal{G}_{x}$  over all $\underline{s}$-leaves $\mathcal{G}_{x}$. The relevant graded bundle  (over $\iota(M) \simeq M $) here is (cf. \cite{Jozwikowski:2014})

 \begin{equation}
 F_{k} = \textnormal{A}^{k}(\mathcal{G}) := \left. \sT^{k}\mathcal{G}^{\underline{s}}\right|_{\iota(M)},
 \end{equation}

 \noindent  which of course inherits its graded  bundle structure as a substructure of $\sT^{k}\mathcal{G}$  with respect to the projection $\zt_{k}:A^k(\mathcal{G})\to M$. For the $k=1$ case this leads to the classical notion of the Lie algebroid associated with the Lie groupoid $\mathcal{G}$. It is clear that directly no Lie bracket structure can exists on the sections of $\textnormal{A}^{k}(\mathcal{G}) \rightarrow M$. Instead it was shown by  J\'{o}\'{z}wikowski \&  Rotkiewicz that there exist the kappa-relation $\sT^{k}\textnormal{A}^{1}(\mathcal{G}) \rRelation \sT \textnormal{A}^{k}(\mathcal{G}) $ satisfying certain properties.  The dual to this relation is the genuine vector bundle morphism $\sT^{*}\textnormal{A}^{k}(\mathcal{G}) \rightarrow \sT^{k}\textnormal{A}^{1}(\mathcal{G})^{*}$ over the anchor map $\textnormal{A}^{k}(\mathcal{G}) \rightarrow \sT^{k}M$. These structure they refer to as \emph{higher order Lie  algebroids}.  This clearly differs from our idea of weighted Lie algebroid structure for $\textnormal{A}^{k}(\mathcal{G})$, which is a canonical structure on the linearisation of $\textnormal{A}^{k}(\mathcal{G})$. This is the fundamental point of divergence of the notion of a weighted Lie algebroid and a higher algebroid associated with a Lie groupoid  in the sense of \cite{Jozwikowski:2014}. Our procedure will be a groupoid version of the reduction described in the previous section.

Using the groupoid version of the law `linearisation of the reduction is the reduction of the linearisation' (Theorem \ref{linearisation-reduction}), we obtain
the linearisation  of $\textnormal{A}^{k}(\mathcal{G})$ as the reduction of $\sT\sT^{k-1}\mathcal{G}^{\underline{s}}$. Here, the first $\sT$ also refers to tangent bundle along fibers of the projection $\tau_{k-1}:\sT^{k-1}\mathcal{G}^{\underline{s}}\to M$,
$\tau_{k-1}(\sT^{k-1}\mathcal{G}_x)=x$.

Using the groupoid action one can obtain a `trivialisation'
$$\sT\sT^{k-1}\mathcal{G}^{\underline{s}}\simeq \{ (V,Z)\in\sT\mathcal{G}^{\underline{s}}\ti_M
\sT \textnormal{A}^{k-1}(\mathcal{G}): \sT\underline{t}(V)=\sT \tau_{k-1}(Z)\}\,.$$
It is clear that invariant vector fields, i.e. section of the linearisation correspond to sections of the above bundle with $\sT\mathcal{G}^{\underline{s}}$ replaced by the Lie algebroid $\textnormal{A}(\mathcal{G})$ of $\mathcal{G}$. In other words,
\be\label{groupoid-linearisation}
D\left(\textnormal{A}^{k}(\mathcal{G})\right)\simeq \{ (Y,Z)\in\textnormal{A}(\mathcal{G})\ti_M
\sT \textnormal{A}^{k-1}(\mathcal{G}): \zr(Y)=\sT\tau_{k-1}(Z)\}\,,
\ee
which is canonically a vector bundle over $\textnormal{A}^{k-1}(\mathcal{G})$ according to the projection
onto the base point of the vector $Z$.
Here, $\zr:\textnormal{A}\to\sT M$ is the anchor of the Lie algebroid $\textnormal{A}$.
The groupoid version of considerations we did in example \ref{Atiyah} shows that $D\left(\textnormal{A}^{k}(\mathcal{G})\right)$ is canonically a weighted algebroid. The module of sections is generated by sections corresponding to $Y$ and $Z$. The first ones
commute as sections of the Lie algebroid $\textnormal{A}$, the second as vector fields. The other
brackets are then determined by the anchor $\zr_k$ which is
$$\zr_k(Y,Z)=Z\,.$$
Lie algebroids like this one are known in the literature under the name of \emph{prolongation of a fibered manifold with respect to a Lie algebroid}  (cf. \cite{Cortes:2009} and references therein). In our case, $D\left(\textnormal{A}^{k}(\mathcal{G})\right)$ is the prolongation of $\textnormal{A}^{k-1}(\mathcal{G})$ with respect to the Lie algebroid $\textnormal{A}(\mathcal{G})$ which, in the notation of \cite{Cortes:2009}, is $\mathcal{T}^{\textnormal{A}(\mathcal{G})}\textnormal{A}^{k-1}(\mathcal{G})$.
\begin{theorem}\label{theom:reductionGroupoid}
For any Lie groupoid $\mathcal{G}$, the linearisation of the reduced graded bundle $\textnormal{A}^k(\mathcal{G})=\sT^k\mathcal{G}/\mathcal{G}$ is canonically a weighted Lie algebroid isomorphic with the prolongation of $\textnormal{A}^{k-1}(\mathcal{G})$ with respect to the Lie algebroid $\textnormal{A}(\mathcal{G})$, $$D\left(\textnormal{A}^{k}(\mathcal{G})\right)\simeq
\mathcal{T}^{\textnormal{A}(\mathcal{G})}\textnormal{A}^{k-1}(\mathcal{G})\,.$$
\end{theorem}

\begin{remark} The weighted Lie algebroid  $D(\textnormal{A}^2(\mathcal{G})) \simeq \mathcal{T}^{\textnormal{A}(\mathcal{G})}\textnormal{A}(\mathcal{G})$ (without reference to the graded structure) appeared in the literature in 2001 in the work of  Mart\'{\i}nez \cite{Martinez:2001} and  Cari\~{n}ena \&  Mart\'{\i}nez \cite{Carinena:2001}, which they referred to as the prolongation of a Lie algebroid. However, their construction is somewhat \emph{ad hoc} and intended as a generalisation of the canonical flip on the double tangent bundle $\sT \sT M$. Their motivation, similar to ours   was to develop Lie algebroid analogues of Lagrangian mechanics. A general prolongation has also been considered by Popescu \& Popescu \cite{Popescu:2001} in a similar context. In retrospect, for integrable Lie groupoids at least, we see that the prolongation of a Lie algebroid has its fundamental geometric origin is in the linearisation of the reduced graded bundle $\textnormal{A}^2(\mathcal{G})=\sT^2\mathcal{G}/\mathcal{G}$.
\end{remark}

The weighted Lie algebroid structure associated with the linearisation and reduction of a higher order tangent bundle of a Lie groupoid  is further explored in \cite{Bruce:2014b} as a means of defining higher order mechanics on a Lie algebroid.  We consider theorem \ref{theom:reductionGroupoid} as being one of the main results of this paper and offers a wealth of natural examples of weighted Lie algebroids.

\medskip

\noindent \textbf{As a homological vector field.} It is clear from the proceeding that $\Pi D(\textnormal{A}^{k}(\mathcal{G}))$ is a subsupermanifold of $\Pi \textnormal{A}(\mathcal{G}) \times_{M} \Pi \sT \textnormal{A}^{k-1}(\mathcal{G})$ essentially defined by the anchor. That is, if we employ homogeneous local coordinates
\begin{equation*}
(\underbrace{x^{A}}_{(0,0)},~ \underbrace{\zx_{1}^{a}}_{(0,1)}, ~\underbrace{y^{b}_{w}}_{(w,0)},~ \underbrace{dx_{1}^{A}}_{(0,1)}, ~ \underbrace{dy_{w+1}^{c}}_{(w,1)}),
\end{equation*}

\noindent on $\Pi \textnormal{A}(\mathcal{G}) \times_{M} \Pi \sT \textnormal{A}^{k-1}(\mathcal{G})$, then $\Pi D(\textnormal{A}^{k}(\mathcal{G}))$ is defined by the condition $dx^{A}_{1}  = \zx_{1}^{a}P_{a}^{A}(x)$. Then one can naturally equip $\Pi D(\textnormal{A}^{k}(\mathcal{G}))$ with homogeneous local coordinates
\begin{equation*}
(\underbrace{x^{A}}_{(0,0)},~ \underbrace{\zx_{1}^{a}}_{(0,1)}, ~\underbrace{y^{b}_{w}}_{(w,0)},~ \underbrace{dy_{w+1}^{c}}_{(w,1)}),
\end{equation*}

\noindent where on has to take into account the anchor map when deducing  the transformation laws for the $dy$'s. In these local coordinates the homological vector field encoding the weighted Lie algebroid structure on $D(\textnormal{A}^{k}(\mathcal{G}))$ is given by
\begin{equation}
Q = \zx_{1}^{a}P_{a}^{A}(x) \frac{\partial}{\partial x^{A}} - \frac{1}{2}\zx_{1}^{a}\zx_{1}^{b}P_{ba}^{c}(x)\frac{\partial}{\partial \zx^{c}_{1}}+ \sum_{w=1}^{k-1} dy_{w+1}^{a} \frac{\partial}{\partial y^{a}_{w}}.
\end{equation}

\section*{Acknowledgments}
The authors  wish to  thank J.~Stasheff and M.~Rotkiewicz for their comments on an earlier draft of this work.

\bibliographystyle{elsarticle-num}

\end{document}